\documentclass[12pt,letter]{article}

\usepackage{tocloft}
\usepackage{cite}
\usepackage{amsmath}
\usepackage{amsthm}
\usepackage{color}
\usepackage{graphicx}
\usepackage[usenames,dvipsnames]{pstricks}
\usepackage{epsfig}
\usepackage{pst-grad} 
\usepackage{pst-plot} 
\usepackage{upgreek} 
\usepackage{mathrsfs}
\usepackage{amssymb}
\usepackage{psfrag}
\usepackage{bbm}
\usepackage{xspace}
\usepackage{dsfont}
\usepackage{cite}
\usepackage{bbding,enumerate}
\usepackage{calc}
\usepackage{graphicx}
\graphicspath{{figs/}{../figs/}}

\theoremstyle{plain}
\newtheorem{theorem}{Theorem}[section]
\newtheorem{proposition}[theorem]{Proposition}
\newtheorem{lemma}[theorem]{Lemma}
\newtheorem{corollary}[theorem]{Corollary}

\newtheorem{problem}[theorem]{Problem}

\theoremstyle{definition}
\newtheorem{remark}[theorem]{Remark}
\newtheorem{definition}[theorem]{Definition}
\newtheorem{example}[theorem]{Example}

\DeclareMathOperator{\Ir}{Ir}
\DeclareMathOperator{\Ieq}{\,\normalfont \rotatebox{90}{\hspace{1.5pt}\small{I}}\,}
\newcommand{\chif}{\bar{\chi}_f}
\newcommand*{\MyDef}{\mathrm{def}}
\newcommand*{\eqdefU}{\ensuremath{\mathop{\overset{\MyDef}{=}}}}
\newcommand*{\eqdef}{\mathop{\overset{\MyDef}{\resizebox{\widthof{\eqdefU}}{\heightof{=}}{=}}}}

\input{mathlig}

\usepackage{xspace}
\usepackage{bbm}
\input{mathlig}

\newcommand{\muspace}{\mspace{1mu}}

\DeclareRobustCommand{\scond}{\mathchoice{\muspace\vert\muspace}{\vert}{\vert}{\vert}}
\mathlig{|}{\scond}

\DeclareRobustCommand{\discint}{\mathchoice{\mspace{-1.5mu}:\mspace{-1.5mu}}{\mspace{-1.5mu}:\mspace{-1.5mu}}{:}{:}}
\mathlig{::}{\discint}
\newcommand{\suchthat}{\colon}

\def\rank{\mathop{\rm rank}\nolimits}%
\newcommand{\Hc}{\mathcal{H}}

\newcommand{\Lc}{\mathcal{L}}

\newcommand{\Yc}{\mathcal{Y}}

\newcommand{\Fcal}{\mathcal{F}}

\newcommand{\Hcal}{\mathcal{H}}

\newcommand{\sh}{{\hat{s}}}

\newcommand{\gt}{{\tilde{g}}}

\def\a{\alpha}
\def\b{\beta}

\def\textiid{i.i.d.\@\xspace}
\newcommand\iid{\ifmmode\text{ i.i.d. } \else \textiid \fi}

\def\clap#1{\hbox to 0pt{\hss#1\hss}}
\def\mathclap{\mathpalette\mathclapinternal}
\def\mathclapinternal#1#2{%
  \clap{$\mathsurround=0pt#1{#2}$}}

\let\oldstackrel\stackrel
\renewcommand{\stackrel}[2]{\oldstackrel{\mathclap{#1}}{#2}}

\begin{document}

\title{Graph Information Ratio}

\author{Lele Wang and Ofer Shayevitz \thanks{This work has been supported by an ERC grant no. 639573, and an ISF grant no. 1367/14. L. Wang is jointly with the Department of Electrical Engineering, Stanford University, Stanford, CA, USA and the Department of Electrical Engineering - Systems, Tel Aviv University, Tel Aviv, Israel (email: wanglele@stanford.edu). O. Shayevitz is with the Department of Electrical Engineering - Systems, Tel Aviv University, Tel Aviv, Israel (email: ofersha@eng.tau.ac.il). Parts of this work appear in the proceedings of the International Symposium on Information Theory (ISIT) 2017.}}

\parskip 3pt

\date{}
\maketitle

\begin{abstract}
We introduce the notion of information ratio $\Ir(H/G)$ between two (simple, undirected) graphs $G$ and $H$, defined as the supremum of ratios $k/n$ such that there exists a mapping between the strong products $G^k$ to $H^n$ that preserves non-adjacency. Operationally speaking, the information ratio is the maximal number of source symbols per channel use that can be reliably sent over a channel with a confusion graph $H$, where reliability is measured w.r.t. a source confusion graph $G$. Various results are provided, including in particular lower and upper bounds on $\Ir(H/G)$ in terms of different graph properties, inequalities and identities for behavior under strong product and disjoint union, relations to graph cores, and notions of graph criticality. Informally speaking, $\Ir(H/G)$ can be interpreted as a measure of similarity between $G$ and $H$. We make this notion precise by introducing the concept of information equivalence between graphs, a more quantitative version of homomorphic equivalence. We then describe a natural partial ordering over the space of information equivalence classes, and endow it with a suitable metric structure that is contractive under the strong product. Various examples and open problems are discussed. 
\end{abstract}

\tableofcontents

\thispagestyle{empty}



\section{Overview of Main Results}
The zero-error capacity of a noisy channel is a well known problem in information theory, originally introduced and studied by Shannon \cite{Shannon1956}. One canonical way to describe this problem is the following: A sender is trying to convey one of $M$ distinct messages to a receiver over a channel with some finite input alphabet $V$. The channel noise is characterized by a (simple, undirected) \textit{channel confusion graph} $H$ over the input alphabet (i.e., with a vertex set $V(H)=V$) and an edge set $E(H)$. The sender maps his messages to a sequence of channel inputs in $v_1,\ldots,v_n$, and the receiver in turn observes an arbitrary sequence of edges $e_1,\ldots,e_n\in E(H)$, such that $v_i\in e_i$. This mapping of messages to inputs is called \textit{zero-error} if the receiver can always determine the message uniquely from the sequence of edges. The rate of the mapping is defined as $R=n^{-1}\log{M}$, which corresponds to the number of bits sent per channel use. The zero-error capacity of the channel, also known as the \textit{Shannon graph capacity} $C(H)$, is the supremum over all rates $R$ for which a zero-error mapping exists.

It is not difficult to verify that the maximal zero-error rate for one use of the channel ($n=1$) is exactly $\log\alpha(H)$,  where $\alpha(H)$ is the independence number of the channel graph $H$. More generally, the graph capacity is given by $C(H) = \log\Theta(H)$, where  
\begin{align*}
  \Theta(H) \eqdef \lim_{n\to\infty} \sqrt[n]{\alpha(H^n)}.
\end{align*}
Here $H^n$ is the $n$-fold strong product (a.k.a. AND product) of $H$ with itself. Recall that for two graphs $H_1,H_2$, the strong product $H_1\boxtimes H_2$ is a graph with vertex set $V(H_1)\times V(H_2)$, where $(h_1,h_2)\sim (h'_1,h'_2)$ if for each $i$ either  $h_i\sim h'_i$ in $H_i$ or $h_i=h'_i$. The limit above exists due to supermultiplicativity, but is in general notoriously difficult to compute or even to approximate \cite{Alon--Lubetzky2006}. 

Another closely related problem is that of zero-error compression. Here we are given a \textit{source confusion graph} $G$, and the sender needs to map \textit{source sequences} $g_1,\ldots,g_k\in G$ to one of $M$ messages, that is then sent noiselessly to the receiver. This mapping is called \textit{zero-error} if any two source sequences $g_1,\ldots,g_k\in G$ and  $g'_1,\ldots,g'_k\in G$ that are not confusable w.r.t. $G$, i.e., for which $g_j\not\sim g'_j$ in at least coordinate $j$, are always mapped to distinct messages. In other words, this means that the receiver is always able to output a list of source sequences that are all confusable with the correct one, and are guaranteed to contain it. The rate of the mapping is $R=k^{-1}\log{M}$, which corresponds to the number of message bits per source symbol. The zero-error compression rate of $G$ is the infimum over all rates for which a zero-error mapping exists. We note that this problem was originally introduced by K\"{o}rner \cite{Korner1973} in a slightly different setting where the source is probabilistic and a small error probability is allowed; in that case, the optimal compression ratio is the so-called \textit{K\"{o}rner entropy} of the graph. 

It is not difficult to verify that the minimal rate for source sequences of length one ($n=1$) is exactly $\log\chi(\overline{G})$,  where $\chi(\overline{G})$ is the chromatic number of the complementary graph $\overline{G}$. More generally, the zero-error graph compression rate of $G$ is given by $\log\chif(G)$, where 
\begin{align*}
  \chif(G) \eqdef \lim_{k\to\infty} \sqrt[k]{\chi(\overline{G}^{\vee k})}.
\end{align*}
Here $G^{\vee k}$ is the $k$-fold OR product of $G$ with itself. Recall that for two graphs $G_1,G_2$, the OR product $G_1\vee G_2$ is a graph with vertex set $V(G_1)\times V(G_2)$, where $(g_1,g_2)\sim (g'_1,g'_2)$ if $g_i\sim g'_i$ in $G_i$ for at least one coordinate $i$. The above limit $\chif(G)$, also known as the \textit{fractional chromatic number} of $\overline{G}$, exists due to supermultiplicativity and can computed by solving a simple linear program \cite{Pirnazar--Ullman2002,Scheinerman--Ullman2011}.   

Now, it is only natural to consider the more general problem where a sender wishes to communicate a source sequence with confusion graph $G$ over a noisy channel with confusion graph $H$. Suppose the sender has a source sequence of length $k$, and can use the channel $n$ times. The sender would like to map the source sequence to the channel inputs in a way that the receiver will always be able to output a list of source sequences that are all confusable with the correct one, and are guaranteed to contain it. In graph theoretic terms, we are looking for a mapping $\phi$ from the vertices of $G^k$ to the vertices of $H^n$, such that for any $g,g'\in G^k$ where $g\not\sim g'$, the images $\phi(g)\not\sim\phi(g')$ in $H^n$. We will call such a mapping \textit{non-adjacency preserving}. Does such a mapping exist? The answer depends on $k$ and $n$. This leads us to define the \textit{information ratio} between $H$ and $G$ to be 
\begin{align*}
\Ir(H/G) \eqdef \sup\left\{k\slash n : \text{$\exists$ a non-adjacency preserving mapping from $G^k$ to $H^n$} \right\}.   
\end{align*}
The information ratio can hence be thought of as the maximal number of source symbols per channel use that can be reliably conveyed to the receiver, in the above sense. In what follows, we study the information ratio and derive many interesting properties, which may be of separate interest beyond the information theoretic motivation. 

A simple lower bound on the information ratio can be obtained by a \textit{separation} scheme, namely where the source is first optimally compressed into a message using $\log\chif(G)$ bits per source symbol, and then this message is optimally sent over the channel using $C(H) = \log\Theta(H)$ bits per channel use. It is straightforward to see that this yields the following lower bound:
\begin{align*}
  \Ir(H/G) \geq \frac{\log\Theta(H)}{\log \chif(G)}.
\end{align*}
This lower bound is sometimes tight. For example, denoting the empty graph over two vertices by $\overline{K_2}$, it is clear by definition that $\Ir(H/\overline{K_2}) = \log\Theta(H)$ and $\Ir(\overline{K_2}/G) = 1/\log\chif(G)$, as these cases reduce to the pure communication and pure compression problems respectively. Noting that $\Theta(\overline{K_2}) = \chif(\overline{K_2})=2$, we see that in these cases separation is (trivially) optimal. Separation can be optimal in richer cases as well. However, and in contrast to the classical JSCC case mentioned above \cite{Shannon1959}, separation is not always optimal. To see that, note $\Ir(G/G)\geq 1$ for any $G$ since we can always map $G$ to $G$ using the identity mapping, which is non-adjacency preserving with $n=k=1$ (in fact, we show this is also the best possible). However, in many cases $\Theta(G)$ is strictly smaller than $\chif(G)$. For example, take $G$ to be the pentagon $C_5$ (cycle over five vertices), where separation only achieves $\frac{\log\Theta(C_5)}{\log \chif(C_5)} = \frac{\log \sqrt{5}}{\log 2.5} = 0.878 < 1$. We thus see that $\Ir(H/G)$ depends in general on the relative structure of the two graphs, which is in some sense what makes this problem rich and interesting.  

We will prove various algebraic identities and inequalities for the information ratio. To that end, one may find it instructive to think of the information ratio, very informally, as 
\begin{align*}
  \stackrel{\text{``}}{\,}\Ir(H/G) \approx \frac{\log{|H|}}{\log{|G|}}\,\stackrel{\text{''}}{\,}
\end{align*}
This statement is meant to imply that, loosely speaking, the information ratio behaves like the ratio of logarithmic graph ``sizes'', in terms of satisfying algebraic identities/inequalities similar to the ones satisfied by positive real numbers, where the strong product $G\boxtimes H$ is thought of as ``multiplication'', and the disjoint union $G+H$ is though of as ``addition''. 

More accurately, we will prove the following relations. First, the product of reciprocal information ratios cannot exceed unity:
\begin{align*}
 \Ir(H/G)\, \Ir(G/H)\le 1. 
\end{align*}
The information ratio is super-multiplicative w.r.t. the strong product:
\begin{align*}
 \Ir(G \boxtimes H/F) \ge \Ir(G/F) + \Ir(H/F),
\end{align*}
and is exactly multiplicative for $F=G$, i.e., 
\begin{align*}
 \Ir(G \boxtimes H/G) = 1+ \Ir(H/G).
\end{align*}
Similarly, it also holds that 
\begin{align*}
 \Ir(F/G \boxtimes H) \ge \frac{\Ir(F/G)\,\Ir(F/H)}{\Ir(F/G) + \Ir(F/H)},
\end{align*}
with equality when $F=G$,
\begin{align*}
   \Ir(G/G\boxtimes H) = \frac{\Ir(G/H)}{1+\Ir(G/H)}.
\end{align*}
Furthermore, the following \textit{information ratio power inequality} holds w.r.t. disjoint union:
\begin{align*}
 \chif(F)^{\Ir(G+H/F)} \ge \chif(F)^{\Ir(G/F)} + \chif(F)^{\Ir(H/F)}  
\end{align*}
This implies that we can informally think of the source graph $F$ as the ``logarithm base''. This inequality can be arbitrarily loose, but is tight e.g. when $G=H=F$, in which case 
\begin{align*}
  \Ir(F+F/ F) = 1+\frac{1}{\log \chif(F)}.  
\end{align*}
In a similar vein, it also holds that 
\begin{align*}
    \Ir(F/F+ F) = \frac{\log\Theta(F)}{1+\log\Theta(F)}.
\end{align*}

A simple but useful observation is that the information ratio between two graphs can be equivalently defined in terms of homomorphisms between the respective complement graphs. A homomorphism $\phi$ from $G$ to $H$ is a mapping of $V(G)$ to (a subset of) $V(H)$ that preserves adjacency, i.e. where $g_1\sim g_2$ in $G$ implies $\phi(g_1)\sim \phi(g_2)$ in $H$. This relation is written $G\to H$. It is therefore immediately clear that $\Ir(H/G)$ is the supermum of $k/n$ such that there exists a homomorphism $\overline{G^k}\to\overline{H^n}$, or equivalently, $\overline{G}^{\vee k}\to\overline{H}^{\vee n}$.  

We use the homomorphic definition in conjunction with some known results on hom-monotone functions, i.e., functions that are monotone w.r.t. homomorphisms, to derive several upper bounds on the information ratio:
\begin{align*}
   \Ir(H/G) \le \min\left\{\frac{\log\Theta(H)}{\log\Theta(G)},\, \frac{\log \vartheta(H)}{\log \vartheta(G)},\, \frac{\log \chif(H)}{\log \chif(G)}\right\},
\end{align*}
where $\vartheta(\cdot)$ is the Lov{\'a}sz theta function \cite{Lovasz1979}, a well known upper bound for graph capacity that can be computed by solving a semidefinite program. We note that $\Theta(G)\leq \vartheta(G)\leq \chif(G)$ always holds, hence for the graph capacity problem it is never interesting to look at the fractional chromatic number. It is thus interesting to note that for the information ratio problem there are cases where the fractional chromatic number bound is better than the Lov{\'a}sz theta function bound. We further provide two additional upper bounds which are not easily computed. The first is given by $\Ir(H/G) \leq \frac{\log\gamma_f(H)}{\log\gamma_f(G)}$, where $\gamma_f(\cdot)$ is a properly defined fractional version of Haemers' min-rank function. The second is given implicitly in terms of a newly defined quantity $\beta_f^{(n)}(G,H)$, which is related to the number of homomorphism between powers of $G$ and $H$.

Our upper bounds can be used in conjunction with specific mappings (e.g., separation/uncoded) to find the exact information ratio in several special cases. For example, separation is clearly tight when either $\Theta(G)=\chif(G)$ or $\Theta(H)=\chif(H)$. Using a simple uncoded mapping and the multiplicativity of the upper bounds, it follows that $\Ir(F^{m_1}/F^{m_2}) = \frac{m_1}{m_2}$ for any graph $F$. This can be generalized to exactly determine $\Ir(F_1^{m_1} + F_1^{m_1}/F_2^{m_2} + F_2^{m_2})$. 

The graph homomorphism approach leads to further interesting observations. We say that two graphs $G$ and $H$ are \textit{homomorphically equivalent}, denoted $G\leftrightarrow H$, if both $G\to H$ and $H\to G$. This induces an equivalence relation on the family of (finite, simple, undirected) graphs. It is well known that for any graph $G$, there exists a unique (up to isomorphism) representative $G^\bullet$ of the equivalence class of $G$, called the \textit{core} of $G$. Loosely speaking, the core is the graph with the least number of vertices in the equivalence class. Examples of cores include complete graphs, odd cycles, and Kenser graphs. Using the notion of a core, we can show that the information ratio $\Ir(G/H)$ depends only on the cores of $\overline{G}$ and $\overline{H}$, hence cores are sufficient for the purpose of computing the information ratio. This can simplify the computation; for example, using the core notion it is readily verified that if $H$ is a disjoint union of $t$ (arbitrary) cliques, and $G$ is a disjoint union of $s$ (arbitrary) cliques, then $\Ir(H/G) = \frac{\log{t}}{\log{s}}$. As another example, it can be shown that if $H_1\to H_2$, then $\Ir(\overline{H_1+H_2}/G) = \Ir(\overline{H_2}/G)$ for any $G$.

Clearly, if $\overline{G}$ and $\overline{H}$ are homomorphically equivalent (i.e., have the same core), then $\Ir(H/G) = \Ir(G/H)=1$. However, the reverse implication does not always hold. For example, if $G = \overline{KG(6,2)}$ and $H = \overline{KG(12,4)}$, where $KG(n,r)$ is the Kneser graph, then $\Ir(G/ H) = \Ir( H/ G) = 1$, but $\overline{H}\not\to\overline{G}$. The reason for this deficiency of the homomorphic equivalence is that it is possible for $\overline{G}$ and $\overline{H}$ not to be homomorphically equivalent, yet for their OR powers to become asymptotically ``almost'' so.

This observation leads us to define a new equivalence relation on graphs: we say that $G$ and $H$ are \textit{information--equivalent}, denoted $G\Ieq H$, if $\Ir(H/G) = \Ir(G/H)=1$. This relation is a coarsening of the homomorphic equivalence of the complement graphs, and captures their asymptotic similarity. Information equivalence turns out to enjoy several nice properties. For instance, the set of information equivalence classes equipped with the function  
\begin{align*}
  d(G,H)\eqdef -\log (\min\{\Ir(G/H),\,\Ir(H/G)\})
\end{align*}
forms a metric space, which is contractive w.r.t. to the strong product, i.e., $d(G\boxtimes F, H\boxtimes F) \le d(G,H)$. Furthermore, information equivalence admits an alternative definition in terms of spectra. We define the \textit{source spectrum} (resp. \textit{channel spectrum}) of a graph $F$ to be the (ordered) set of all ratios $\{\Ir(H/F)\}$ (resp. $\{\Ir(F/G)\}$) where $F$ serves as the source (resp. channel). Then two graphs are information--equivalent if and only if they have the same source (resp. channel) spectrum.  

There exists a natural \textit{information partial order} $\preccurlyeq$ over the set of information equivalence classes, defined by $G \preccurlyeq H$ if $\Ir(H/G) \ge 1$. Not all graphs are comparable w.r.t. the information partial order; in the sequel we give an example of graphs such that both $\Ir(G/H) < 1$ and  $\Ir(H/G) < 1$. We show that $G \preccurlyeq H$ if and only if the source (resp. channel) spectrum of $G$ is point-wise not smaller (resp. not larger) than the source (resp. channel) spectrum of $H$. Furthermore, the functions $\Theta(G), \vartheta(G), \chif(G)$ and $\gamma_f(G)$ are all monotonically non-decreasing w.r.t. the information partial order. 

We further say that $G$ and $H$ are \textit{weakly information--equivalent}, denoted $G\Ieq_w H$, if $\Ir(G/H)\Ir(H/G) = 1$. This is a coarsening of the information equivalence that is insensitive to strong products, in the sense that now $G^k\Ieq_w G^n$ for any $k,n$. The set of associated equivalence classes equipped with the function $d_w(G,H)\eqdef -\log (\Ir(G/H)\Ir(H/G))$ forms a metric space. Moreover, the following identity holds
\begin{align*}
  \Ir(G\boxtimes H/ F) = \Ir(G/ F) + \Ir(H/ F),  
\end{align*}
whenever either $G \Ieq_w H$ or $F \Ieq_w G$ or $F \Ieq_w H$.

Finally, we discuss a notion of graph criticality induced by the information ratio. A graph $F$ is called \textit{information--critical} if there exists an edge $e\in E(F)$ that, once removed, changes all the information ratios, i.e., $\Ir(H/(F\setminus e)) < \Ir(H/F)$ and $\Ir((F\setminus e)/G) > \Ir(F/G)$ for all $G,H$. We provide equivalent characterizations as well as sufficient conditions for information--criticality, and show that complements of various known cores are information--critical. 

\subsubsection*{Related work}
The question we consider is reminiscent of the \textit{joint source-channel coding (JSCC)} problem studied in the classical (non-zero-error) information theoretic setup, where a source is to be communicated over a noisy channel under some end-to-end distortion constraint (expected distortion, excess distortion exponent, etc.)~\cite{Shannon1959}. Our problem differs from these classical JSCC setups in a number of aspects. First, our setting is combinatorial in nature and does not allow any errors; in this sense, a more closely related study appears in two papers by Kochman \textit{et al} \cite{Kochman--Mazumdar--Polyanskiy2012a,Kochman--Mazumdar--Polyanskiy2012b} where the authors consider JSCC over an adversarial channel. But more importantly, the way we measure success cannot be cast in a per-symbol distortion JSCC framework. The natural distortion for our setup is one where confusable symbols are assigned zero distortion and non-confusable symbols are assigned infinite distortion. However, this results in a different (more degenerate) setup; for example, if $G$ has a vertex that is connected to all the other vertices, then the receiver can always reconstruct this vertex with no communication cost, and admit zero distortion. A better way to think about our setup is perhaps that of list decoding with structural constraints. Unlike the pure communication setup where the receiver must commit to one message, here the receiver can output a list of possible messages, but this list must have a certain ``similarity structure'' that is captured by the source confusion graph $G$. 

In a related work~\cite{Korner--Marton2001}, K\"{o}rner and Marton introduced and studied the \emph{relative capacity} $R(H|G)$ of a graph pair $(G,H)$, defined (originally in terms of OR products) as the maximal ratio $k/n$ such that $H^n$ contains an induced subgraph that is isomorphic to $G^k$. This is a stronger requirement than ours; the information-ratio setup is concerned with mappings (from $G^k$ to $H^n$) that are only required to preserve non-adjacency, whereas the relative capacity setup considers mappings that must preserve both adjacency and non-adjacency. The relative capacity $R(H|G)$ is therefore a lower bound on the information ratio $\Ir(H/G)$. We note that the relative capacity was determined in~\cite{Korner--Marton2001} for the special case where $G$ is an empty graph, where it can be easily seen to equal the associated information ratio. The main contribution in~\cite{Korner--Marton2001} is a general upper bound on the relative capacity, and its ramifications in a certain problem of graph dimension. This upper bound is not informative in the associated information ratio setup. 

In another related work~\cite{Polyanskiy2017}, Polyanskiy studied certain mappings of the Hamming graph $H(m,d)$, which is a graph over the Hamming cube $\{0,1\}^m$ where two vertices are connected if their Hamming distance is at most $d$. He investigated conditions for the non-existence of  $(\alpha,\beta)$-mappings, which he defined as non-adjacency preserving mappings between $H(m,\alpha m)$ and $H(\ell,\beta \ell)$ (for a fixed $n=k=1$ in our notation, i.e., no graph products). He then provided impossibility results in the limit of $m\to\infty$ for a fixed ratio $m/\ell$, via a derivation of general conditions for existence of graph homomorphism. This problem is also closely related to the combinatorial JSCC setup mentioned above. 

In a recent work \cite{Fritz2017}, Fritz has independently studied general problems of resource conversion in an abstract category theory framework. When interpreted in the case of the ordered commutative monoid of graphs, some his abstract results can be linked to our work.

\subsubsection*{Notation}
We typically reserve $G$ for a source graph and $H$ for a channel graph, hence a pair of graphs $(G,H)$ is a source-channel pair in this order. For such a pair, a non-adjacency preserving mapping from $G^k$ to $H^n$ is called a \textit{$(k,n)$ code}. We denote by $\overline{G}$ the complement  of the graph $G$, i.e., $g_1 \sim g_2$ in $\overline{G}$ if and only if $g_1 \not\sim g_2$ in $G$.  We write $K_s$ to mean a complete graph over $s$ vertices, and $\overline{K_s}$ to mean the empty graph over $s$ vertices. We write $C_n$ to mean a cycle with $n$ vertices, and $W_n$ to mean a wheel, which is a cycle $C_n$ with an extra vertex that is adjacent to all vertices in the cycle. We write $KG(n,r)$ to mean the Kneser graph, whose vertices are all $r$-subsets of $\{1,2,\ldots,n\}$, and where two vertices are adjacent if and only if they correspond to disjoint subsets.

\section{Lower Bounds}

In this section, we derive several lower bounds on the information ratio. We first recall the two extremes of the information ratio problem, which follow directly from the definitions. 
\begin{proposition}[Zero-Error Channel Coding]
\label{prop:cap}
 Let $G = \overline{K_2}$ and $H$ be any graph. Then
 \[
  \Ir(H/\overline{K_2}) = \log\Theta(H).
 \]
\end{proposition}

\begin{proposition}[Zero-Error Source Compression]
\label{prop:chif}
Let $G$ be any graph and $H = \overline{K_2}$. Then
\[
 \Ir(\overline{K_2}/G) = \frac{1}{\log \chif(G)}.
\]
\end{proposition}

The following simple lemma is useful in the sequel. 
\begin{lemma}
\label{lem:m-ext}
If a $(k,n)$ code exists for $(G,H)$, then a $(km,nm)$ code exists for $(G,H)$, for any $m \in \mathbbm{N}$.
\end{lemma}
\begin{proof}
Given any non-adjacency preserving mapping $f\suchthat V(G^k) \to V(H^n)$ for the source-channel pair $(G,H)$, the mapping $(g_1,\ldots,g_m) \mapsto (f(g_1),\ldots,f(g_m))$ is again non-adjacency preserving. This leads to a $(km,nm)$ code for the source-channel pair $(G,H)$.
\end{proof}

Next, we show how to construct a code for the pair $(G,H)$ by concatenating two codes for the pairs $(G,F)$ and $(F,H)$.
\begin{lemma}[Concatenation Scheme]
\label{lem:concat}
For any graphs $G,H,F$,
\[
 \Ir(H/G) \ge \Ir(F/G)\Ir(H/F)
\]
\end{lemma}

\begin{proof}
Applying Lemma~\ref{lem:m-ext}, given a $(k_1,n_1)$ code for the pair $(G,F)$ and a $(k_2,n_2)$ code for the pair $(F,H)$, we can construct a $(k_1k_2,n_1k_2)$ code for the pair $(G,F)$, and an $(n_1k_2,n_1n_2)$ code for the pair $(F,H)$. Denote by $f\suchthat V(G^{k_1k_2}) \to V(F^{n_1k_2})$ and $h\suchthat V(F^{n_1k_2}) \to V(H^{n_1n_2})$ the non-adjacency preserving mappings for the pairs $(G,F)$ and $(F,H)$ respectively. Then, the composition code $h\circ f\suchthat V(G^{k_1k_2}) \to V(H^{n_1n_2})$ is non-adjacency preserving for the pair $(G,H)$. Thus, $\Ir(H/G) \ge \frac{k_1k_2}{n_1n_2}$. Maximizing over all codes for the pairs $(G,F)$ and $(F,H)$ completes the proof.
\end{proof}

A special case of the concatenation scheme leads to the following lower bound on the information ratio.

\begin{theorem}[Separation Scheme] 
\label{thm:separation}
For a source-channel pair $(G,H)$, 
\[
  \Ir(H/G) \ge \frac{\log\Theta(H)}{\log \chif(G)}.
\]
\end{theorem}

\begin{proof}
Combine Propositions~\ref{prop:cap},~\ref{prop:chif}, and~\ref{lem:concat} with $F = \overline{K_2}$ in Lemma~\ref{lem:concat}.
\end{proof}

\begin{remark}
 In the classical probabilistic joint source-channel coding, the separation scheme that compresses the source using the optimal source code and transmits it over the noisy channel using the optimal channel code turns out to be the best possible~\cite{Shannon1959}. However, for our problem, which can be viewed as a zero-error joint source-channel coding problem, the separation scheme can be strictly suboptimal. We illustrate this in the following example. 
\end{remark}

\begin{example}[Separation can be strictly suboptimal]
 Let $G = H = C_5$ be the pentagon (a cycle with five vertices). Clearly, the ``uncoded'' (identity) mapping from $G$ to $H$ is non-adjacency preserving. Thus $\Ir(H/G) \ge 1$. However, the separation scheme only achieves $\frac{\log\Theta(C_5)}{\log \chif(C_5)} = \frac{\log \sqrt{5}}{\log 2.5} = 0.878 < 1$.
\end{example}

We now provide lower bounds for the case where either the source graph or the channel graph is a strong product or a disjoint union. 

\begin{remark}
  For brevity of exposition, we assume throughout our proofs that the information ratios are achieved for finite $(k,n)$. This is of course not necessarily the case, but can be easily dealt with by taking appropriate limits in a trivial way. Furthermore, many of our statements should be understood for $n,k$ large enough, as will be clear from the context. 
\end{remark}

\begin{theorem}[{\bf Information ratio product inequality}]
\label{thm:Ir-prod}
Let $G,H,F$ be graphs. Then
\begin{equation}
 \label{eqn:Ir-prod}
 \Ir(G\boxtimes H/F) \ge \Ir(G/F) + \Ir(H/F).
\end{equation}
\end{theorem}

\begin{remark}
 Recall the standard definition of graph capacity $C(H) = \log\Theta(H)$. By Proposition~\ref{prop:cap}, when $F = \overline{K_2}$, Theorem~\ref{thm:Ir-prod} recovers Shannon's lower bound on capacity of the strong product of two graphs
\begin{equation}
\label{eqn:Shannon-prod}
  C(G\boxtimes H) \ge C(G) +  C(H).
\end{equation}
Shannon conjectured that equality in~\eqref{eqn:Shannon-prod} holds~\cite{Shannon1956}, which was then disproved by Alon~\cite{Alon1998}. There are graphs for which the inequality can be strict in~\eqref{eqn:Shannon-prod}, and hence also in~\eqref{eqn:Ir-prod}. In Theorem~\ref{thm:gh-prod} and Theorem~\ref{thm:sum-ratio-tight}, we discuss conditions under which equality in~\eqref{eqn:Ir-prod} holds.
\end{remark}

\begin{proof}
We construct a mapping from $F^{k_1}\boxtimes F^{k_2}$ to $G^n \boxtimes H^n$ as follows. Let $k_1/n = \Ir(G/F)$ and map $F^{k_1}$ to $G^n$ using an optimal code for the pair $(F,G)$. Similarly let $k_2/n = \Ir(H/F)$ and map $F^{k_2}$ to $H^n$ using an optimal code for the pair $(F,H)$. Thus, the ratio $(k_1+k_2)/n = \Ir(G/F)+\Ir(H/F)$ is achievable for the pair $(F,G\boxtimes H)$.
\end{proof}

\begin{theorem}[{\bf Information ratio reverse product inequality}]
 \label{thm:Ir-invprod}
 Let $G,H,F$ be graphs. Then,
 \begin{equation}
  \label{eqn:Ir-invprod}
  \Ir(F/G\boxtimes H) \ge \frac{\Ir(F/G)\Ir(F/H)}{\Ir(F/G)+\Ir(F/H)}.
 \end{equation}
\end{theorem}

\begin{remark}
 By Proposition~\ref{prop:chif}, when $F = \overline{K_2}$, Theorem~\ref{thm:Ir-invprod} holds with equality. This is the well-known fact that $\chif(G\boxtimes H) = \chif(G)\chif(H)$. In Theorem~\ref{thm:gh-prod} and Theorem~\ref{thm:inv-ratio-tight}, we discuss other conditions under which equality in~\eqref{eqn:Ir-invprod} holds. 
\end{remark}

\begin{proof}
We construct a mapping from $G^{k}\boxtimes H^k$ to $F^{n_1} \boxtimes F^{n_2}$ as follows. Let $k/n_1 = \Ir(F/G)$ and map $G^{k}$ to $F^{n_1}$ using an optimal code for the pair $(G,F)$. Similarly, let $k/n_2 = \Ir(F/H)$ and map $H^k$ to $F^{n_2}$ using an optimal code for the pair $(H,F)$. Thus, the ratio $\frac{k}{n_1+n_2} = \frac{\Ir(F/G)\Ir(F/H)}{\Ir(F/G)+\Ir(F/H)}$ is achievable for the pair $(G\boxtimes H,F)$.
\end{proof}

\begin{theorem}[{\bf Information ratio power inequality}]
 \label{thm:union}
 Let $G,H,F$ be graphs. Then,
\begin{equation}
\label{eqn:sum-ratio}
 \chif(F)^{\Ir(G+H/F)} \ge \chif(F)^{\Ir(G/F)} + \chif(F)^{\Ir(H/F)}.
\end{equation}
\end{theorem}

\begin{remark}
 Recall the standard definition of graph capacity $C(H) = \log\Theta(H)$. By Proposition~\ref{prop:chif}, when $F = \overline{K_2}$, Theorem~\ref{thm:union} recovers Shannon's lower bound on the capacity of disjoint union of two graphs
 \begin{equation}
 \label{eqn:sum-cap}
  2^{C(G+H)} \ge 2^{C(G)} + 2^{C(H)}.
 \end{equation}
Shannon conjectured that equality in~\eqref{eqn:sum-cap} holds~\cite{Shannon1956}, which was then disproved by Alon~\cite{Alon1998}. There are graphs for which the inequality in~\eqref{eqn:sum-cap} can be strict, hence also in~\eqref{eqn:sum-ratio}. 
\end{remark}

\begin{proof}
Note that $(G+H)^n \cong \sum_{i=0}^n \binom{n}{i} G^i\boxtimes H^{n-i}$. Let $i = \a n$. We consider a mapping from $F^{k_1}\boxtimes F^{k_2}\boxtimes F^{k_3}$ to $\binom{n}{\a n} G^{\a n}\boxtimes H^{(1-\a)n}$. We map $F^{k_1}$ to an empty graph with $\binom{n}{\a n}$ vertices. This can be done if $k_1\log\chif(G) \leq   \log\binom{n}{\a n}\approx nh(\a)$, where $h(\a) = -\a\log\a-(1-\a)\log(1-\a)$ is the binary entropy function. We set $\frac{k_2}{\a n} = \Ir(G/F)$ and map $F^{k_2}$ to $G^{\a n}$ with the optimal ratio. We set $\frac{k_3}{(1-\a)n} = \Ir(H/F)$ and map $F^{k_3}$ to $H^{(1-\a)n}$ with the optimal ratio. In summary, we have 
\[
\Ir(G+H/F) \ge \frac{k_1+k_2+k_3}{n} = \frac{h(\a)}{\log\chif(F)} + \a \Ir(G/F) + (1-\a)\Ir(H/F)
\]
for any $\a \in [0,1]$. Taking derivative with respect to $\alpha$, we obtain the $\alpha$ that maximizes the lower bound:
\[
 \a^\ast = \frac{\chif(F)^{\Ir(G/F)}}{\chif(F)^{\Ir(G/F)} + \chif(F)^{\Ir(H/F)}}.
\]
Plugging it in, we have
\begin{align*}
\frac{h(\a^\ast)}{\log\chif(F)}&= \frac{-\a^\ast \log (\a^\ast) - (1-\a^\ast)\log (1-\a^\ast)}{\log\chif(F)} \\
&\stackrel{(a)}{=} -\a^\ast \log_{\chif(F)} (\a^{\ast}) - (1-\a^\ast) \log_{\chif(F)} (1-\a^\ast)\\
&= -\a^\ast \log_{\chif(F)}\left(\frac{\chif(F)^{\Ir(G/F)}}{\chif(F)^{\Ir(G/F)} + \chif(F)^{\Ir(H/F)}}\right)\\
&\hspace{1.2em}-(1-\a^\ast) \log_{\chif(F)} \left(\frac{\chif(F)^{\Ir(H/F)}}{\chif(F)^{\Ir(G/F)} + \chif(F)^{\Ir(H/F)}}\right)\\
&= -\a^\ast \Ir(G/F) -(1-\a^\ast)\Ir(H/F)+\log_{\chif(F)} \left(\chif(F)^{\Ir(G/F)} + \chif(F)^{\Ir(H/F)}\right),
\end{align*}
where in $(a)$, we change the base of the logarithm from $2$ to $\chif(F)$. Thus,
\begin{align*}
 \Ir(G+H/F)  &\ge \frac{h(\a^\ast)}{\log\chif(F)} + \a^\ast \Ir(G/F) + (1-\a^\ast)\Ir(H/F)\\
 &= \log_{\chif(F)} \left(\chif(F)^{\Ir(G/F)} + \chif(F)^{\Ir(H/F)}\right).
\end{align*}
Finally, noting that $\chif(F) \ge 1$ completes the proof. 
\end{proof}

\section{Identities}
In this short section, we provide a few information ratio identities. 
\begin{theorem}
\label{thm:gh-prod}
For any graphs $G$ and $H$,
\begin{align}
 \Ir(G \boxtimes H/G) &= 1+\Ir(H/G), \label{eqn:prod1}\\
 \Ir(G/G\boxtimes H) &= \frac{\Ir(G/H)}{1+\Ir(G/H)}. \label{eqn:prod2}
\end{align}
\end{theorem}

\begin{corollary}
\label{cor:gh-prod}
For any graphs $G$ and $H$,
 \[
  \Ir(H/G) = \Ir(G\boxtimes H/G) \Ir(H/G\boxtimes H).
 \]
\end{corollary}

\begin{proof}
Taking $F = G$ in~\eqref{eqn:Ir-prod} and noting that the identity mapping from $G$ to $G$ is non-adjacency preserving, we have
\begin{equation}
\label{eqn:prod1-ge}
 \Ir(G\boxtimes H/G) \ge \Ir(G/G) + \Ir(H/G) \ge 1 + \Ir(H/G).
\end{equation}

Now, let us show that
\begin{equation}
 \label{eqn:prod2-ge}
 \Ir(H/G\boxtimes H) \ge \frac{\Ir(H/G)}{1+\Ir(H/G)}.
\end{equation}
Consider a mapping from $G^k \boxtimes H^k$ to $H^{n_1} \boxtimes H^{n_2}$. We set $k/{n_1} = \Ir(H/G)$ and map $G^k$ to $H^{n_1}$ using an optimal code for the pair $(G,H)$. We set $k = n_2$ and map $H^k$ to $H^{n_2}$ with the identity mapping. Thus the ratio $k/(n_1+n_2) = \frac{\Ir(H/G)}{1+\Ir(H/G)}$ is achievable for the source channel pair $(G\boxtimes H, H)$.

Finally, combining~\eqref{eqn:prod1-ge} and~\eqref{eqn:prod2-ge}, we obtain
\begin{align*}
 \Ir(H/G) &\stackrel{(a)}{\ge} \Ir(G\boxtimes H/G) \Ir(H/G\boxtimes H)\\
 &\ge (1 + \Ir(H/G))\left(\frac{\Ir(H/G)}{1+\Ir(H/G)}\right)\\
 &= \Ir(H/G),
\end{align*}
where $(a)$ follows by virtue of Lemma~\ref{lem:concat} with $F = G\boxtimes H$. Clearly, it must be that both inequalities above hold with equality. In particular,~\eqref{eqn:prod1-ge} holds with equality, which is exactly~\eqref{eqn:prod1}. Also,~\eqref{eqn:prod2-ge} holds with equality, which yields~\eqref{eqn:prod2}. Finally, the equality in $(a)$ establishes Corollary~\ref{cor:gh-prod}.
\end{proof}

\section{Upper Bounds}
\label{sec:upper-bd}

In this section, we study upper bounds on the information ratio. To that end, we reformulate the information ratio problem in the language of graph homomorphisms.  
 
\subsection{General Graph Pairs}

Let $G = (V(G),E(G))$ and $H = (V(H),E(H))$ be two graphs. A \emph{graph homomorphism} from $G$ to $H$, written as $G \to H$, is a mapping $f\suchthat V(G) \to V(H)$ such that $f(g_1)\sim f(g_2)$ in $H$ whenever $g_1 \sim g_2$ in $G$. Recall that $G_1 \vee G_2$ denotes the OR product of $G_1$ and $G_2$, i.e., two vertices $(g_1,g_2)\sim (g_1',g_2')$ are connected if $g_1 \sim g_1'$ or $g_2 \sim g_2'$. We denote by $G^{\vee k}$ the $k$-fold OR product of $G$. By definition, we have $\overline{G\boxtimes H} = \overline{G}\vee \overline{H}$. Since graph homomorphism is the complementary notion of a non-adjacency preserving mapping, the information ratio problem can be trivially reformulated as a graph homomorphism existence questions, as spelled out in the following simple lemma.  
\begin{lemma}
\label{lem:graph-hom}
 A $(k,n)$ code for the pair $(G,H)$ exists if and only if there exists a graph homomorphism $\overline{G^k} \to \overline{H^n}$, or equivalently $\overline{G}^{\vee k} \to \overline{H}^{\vee n}$. 
\end{lemma}

\begin{proof}
 A $(k,n)$ code exists if and only if there is a mapping $f\suchthat V(G^k) \to V(H^n)$ such that $f(x^k) \nsim f(y^k)$ in $H^n$ whenever $x^k \nsim y^k$ in $G^k$. This is equivalent as $f(x^k) \sim f(y^k)$ in $\overline{H^n}$ whenever $x^k \sim y^k$ in $\overline{G^k}$, which is definition of a graph homomorphism $\overline{G^k} \to \overline{H^n}$.
\end{proof}

Next, we recall several well-known necessary conditions for the existence of graph homomorphisms. 

\begin{lemma}[Hom-monotone functions]
\label{lem:hom-mono}
 If there exists a graph homomorphism $G \to H$, then
 \begin{enumerate}
  \item the independence numbers satisfy $\a(\overline{G}) \le \a(\overline{H})$~\cite{Godsil--Royle2001};
  \item the Lov{\'a}sz theta functions~\cite{Lovasz1979} satisfy $\vartheta(\overline{G}) \le \vartheta(\overline{H})$~\cite[Section 4]{de-Carli-Silva--Tuncel2013};
  \item the chromatic numbers satisfy $\chi(G) \le \chi(H)$~\cite{Godsil--Royle2001}.
 \end{enumerate}
\end{lemma}

We now use the above relations to obtain upper bounds on the information ratio. Recall the notation $\chif(H) \eqdef \chi_f(\overline{H})$.
\begin{theorem}[Hom-monotone upper bounds]
For a pair $(G,H)$, 
\label{thm:meta-upb}
\[
 \Ir(H/G) \le \min\left\{\frac{\log\Theta(H)}{\log\Theta(G)},\, \frac{\log \vartheta(H)}{\log \vartheta(G)},\, \frac{\log \chif(H)}{\log \chif(G)}\right\}.
\]
\end{theorem}

Applying any one of the upper bounds twice, we get the following corollary, which will be useful later.

\begin{corollary}
 \label{cor:meta-upb}
 For any graphs $G$ and $H$, 
\begin{equation}
\label{eqn:reciprocal}
\Ir(H/G)\Ir(G/H) \le 1.
\end{equation}
\end{corollary}


\begin{proof}[Proof of Theorem~\ref{thm:meta-upb}]
Suppose that there exists a $(k,n)$ code for the pair $(G,H)$. By Lemma~\ref{lem:m-ext}, for any positive integer $m$, a $(km,nm)$ code exists. Applying Lemma~\ref{lem:graph-hom}, there exists a graph homomorphism $\overline{G^{km}} \to \overline{H^{nm}}$. It follows from Lemma~\ref{lem:hom-mono} that
 \begin{align*}
  \a(G^{km}) &\le \a(H^{nm}),\\
  \vartheta(G^{km}) &\le \vartheta(H^{nm}),\\
  \chi(\overline{G^{km}}) &\le \chi(\overline{H^{nm}}).
 \end{align*}
equivalently, we have 
 \begin{align*}
  \tfrac{1}{n} \left(\tfrac{1}{km}\log\a(G^{km})\right) &\le \tfrac 1k \left(\tfrac{1}{nm}\log \a(H^{nm})\right),\\
  \tfrac{1}{n} \left(\tfrac{1}{km} \log\vartheta(G^{km})\right) &\le \tfrac 1k \left(\tfrac{1}{nm}\log \vartheta(H^{nm})\right),\\
  \tfrac{1}{n} \left(\tfrac{1}{km} \log\chi(\overline{G}^{\vee km})\right) &\le \tfrac 1k \left(\tfrac{1}{nm}\log \chi(\overline{H}^{\vee nm})\right).
 \end{align*}
Now recall that $\log\Theta(G) = \lim_{m \to \infty}\tfrac{1}{m}\log \a(G^m)$, $\log\vartheta(G) = \tfrac{1}{m}\log\vartheta(G^m)$ for any $m$, and that $\log\chi_f(G) = \lim_{m \to \infty}\tfrac{1}{m}\log\chi(G^{\vee m})$. Letting $m \to \infty$ while keeping the ratio $k/n$ a constant, we establish Theorem~\ref{thm:meta-upb}.
\end{proof}

Unlike the case of a single graph, in which there is an order among the three graph invariants $\Theta(G) \le \vartheta(G) \le  \chif(G)$, there is in general no order among the three upper bounds for information ratio, as we now demonstrate. 

\begin{example}[No orders among the upper bounds]
\label{ex:Schlafli}
 Let $G$ be a \emph{strongly regular graph} with parameter $(27,16,10,8)$, i.e., a graph with 27 vertices such that every vertex has 16 neighbors, every adjacent pair of vertices has 10 common neighbors, and every nonadjacent pair has 8 common neighbors~\cite[pp. 464--465]{West2001}. This is also called the Schl\"{a}fli graph~\cite{Schlafli-graph}. It is known that $\Theta(G) = 3$, $\vartheta(G) = 3$, $\chif(G) = 4.5$, $6 \le \Theta(\overline{G}) \le  7$, and $\vartheta(\overline{G}) = \chif(\overline{G}) = 9$~\cite{Schlafli-graph,Haemers1979}.
 \begin{enumerate}
  \item For the pair $(\overline{K_2},\overline{G})$, the upper bound in terms of capacity  is the tightest:
 \[
  \frac{\log\Theta(\overline{G})}{\log\Theta(\overline{K_2})} \le \log 7, \quad \frac{\log\vartheta(\overline{G})}{\log\vartheta(\overline{K_2})} = \frac{\log\chif(\overline{G})}{\log\chif(\overline{K_2})} = \log 9.
 \]
 \item For the pair $(\overline{G},G)$, the upper bound in terms of Lov{\'a}sz's theta function is the tightest:
 \[
  \frac{\log\Theta(G)}{\log\Theta(\overline{G})} \ge \frac{\log 3}{\log 7} = 0.56, \quad \frac{\log\vartheta(G)}{\log\vartheta(\overline{G})} = \frac{\log 3}{\log 9} = 0.5, \quad \frac{\log\chif(G)}{\log\chif(\overline{G})} = \frac{\log 4.5}{\log 9} = 0.68.
 \]
 \item For the pair $(G,\overline{G})$, the upper bound in terms of fractional chromatic number is the tightest:
 \[
  \frac{\log\Theta(\overline{G})}{\log\Theta(G)} \ge \frac{\log 6}{\log 3} = 1.63, \quad \frac{\log\vartheta(\overline{G})}{\log\vartheta(G)} = \frac{\log 9}{\log 3} = 2, \quad \frac{\log\chif(\overline{G})}{\log\chif(G)} = \frac{\log 9}{\log 4.5} = 1.46. 
 \]
\end{enumerate}
\end{example}

In the remainder of this subsection, we derive another upper bound on the information ratio in terms of Haemers' minrank function.

\begin{definition}[Haemers' minrank function~\cite{Haemers1978}]
Let $F = (V,E)$ be a (simple, undirected)  graph with $m$ vertices. We say that an $m\times m$ matrix $B$ over a field $\mathbb{F}$ \emph{fits} $F$ if the following two conditions hold:
\begin{enumerate}
 \item $B_{ii} \neq 0$ for any $i \in V$;
 \item $B_{ij} = 0$ if $i\not\sim j$, for any $i\neq j\in V$. 
\end{enumerate}
Haemers' minrank function (for the field $\mathbb{F}$) is defined as
\[
 \gamma(F) \eqdef \min\{\rank(B)\suchthat \text{the matrix } B \text{ fits } F\}.
\]
\end{definition}

\begin{lemma}[minrank is hom-monotone]
 \label{lem:Haemers}
 If there exists a graph homomorphism $X \to Y$, then
\begin{equation}
\label{eqn:Haemers}
  \gamma(\overline{X}) \le \gamma(\overline{Y}).
\end{equation}
\end{lemma}

\begin{proof}
We first introduce Alon's representation of a graph using polynomials~\cite{Alon1998}, and show it is an equivalent way to describe Haemers' minrank function. Then, we prove Lemma~\ref{lem:Haemers} using this representation.

Let $F=(V,E)$ be a graph and let $\Fcal$ be a subspace of the space of polynomials in $r$ variables over the field $\mathbb{F}$. A \emph{representation} of $F$ over $\Fcal$ is an assignment of a polynomial $f_i$ in $\Fcal$ to each vertex $i \in V$ and an assignment of a point $c_i \in \mathbb{F}^r$ to each $i \in V$ such that the following two conditions hold:
\begin{enumerate}
 \item $f_i(c_i) \neq 0$ for any $i \in V$;
 \item $f_i(c_j) = 0$ if $i\not\sim j$, for any $i\neq j\in V$. 
\end{enumerate}
First, we show that
\begin{equation}
\label{eqn:equiv}
\gamma(F) = \min\{\dim(\Fcal)\suchthat F \text{ has a representation over } \Fcal\}.
\end{equation}

Given a matrix $B$ that fits $F$, we let $f_i$ be a polynomial in $|V|$ variables whose coefficients are given by the $i$-th row of $B$ for each $i \in V$. Let $c_i = (0,\ldots,0,1,0,\ldots,0)$, where $1$ is at the $i$-th coordinate. Then, $f_i(c_i) = B_{ii} \neq 0$ for all $i \in V$, and $f_i(c_j) = B_{ij} = 0$ if $i$ and $j$ are distinct non-adjacent vertices of $G$. Thus, $\{f_i,c_i\suchthat i\in V\}$ is a representation of $F$. Clearly $\{f_i\suchthat i\in S\subseteq V\}$ are linearly independent whenever the corresponding rows of $B$ are linearly independent. Hence $\min\{\dim(\Fcal)\suchthat F \text{ has a representation over } \Fcal\}\le \gamma(F)$.   

Conversely, let $\{f_i,c_i\suchthat i\in V\}$ be a representation of $F$ over $\Fcal$, and set $B_{ij} = f_i(c_j)$. Then $B_{ii} = f_i(c_i) \neq 0$ for all $i \in V$, and $B_{ij} = f_i(c_j) = 0$ if $i$ and $j$ are distinct non-adjacent vertices of $F$. Thus, the matrix $B$ fits $F$. If $\{f_i\suchthat i \in S\subseteq V\}$ are linearly independent, then the corresponding rows in $B$ are linearly independent. It follows that $\min\{\dim(\Fcal)\suchthat F \text{ has a representation over } \Fcal\} \geq \gamma(F)$. This completes the proof of~\eqref{eqn:equiv}.

Now we are ready to prove~\eqref{eqn:Haemers}. Suppose that $\phi$ is a homomorphism $X \to Y$. Let $\{f_i,c_i\suchthat i\in V_Y\}$ be a representation of $\overline{Y}$ over $\Fcal$. We construct a representation of $\overline{X}$ as follows. For each $i \in V_X$, let $h_{i} = f_{\phi(i)}$ and $d_i = c_{\phi(i)}$. Then, for each $i \in V_X$, $h_i(d_i) = f_{\phi(i)}(c_{\phi(i)}) \neq 0$. If $i$ and $j$ are distinct non-adjacent vertices of $\overline{X}$, then $i\sim j$ in $X$. Since $\phi$ is a homomorphism $X \to Y$, we have $\phi(i)\sim \phi(j)$ in $Y$ and thus $\phi(i)\not\sim \phi(j)$ in $\overline{Y}$. It follows that $h_i(d_j) = f_{\phi(i)}(c_{\phi(j)}) = 0$. Therefore, $\{h_i,d_i\suchthat i \in V_X\}$ is a representation of $\overline{X}$. Clearly $\{h_i\suchthat i\in S\subseteq V_X\}$  are linearly independent whenever $\{f_{\phi(i)}\suchthat i \in S\}$ are linearly independent. Therefore, the dimension of the subspace containing $\{h_i\suchthat i \in V_X\}$ is smaller or equal to $\dim(\Fcal)$. It follows that 
$\min\{\dim(\Hcal)\suchthat \overline{X} \text{ has a representation over } \Hcal\}$ is at most $\min\{\dim(\Fcal)\suchthat \overline{Y} \text{ has a representation over } \Fcal\}$.
\end{proof}

It is easy to check that if $A$ fits $G$ and $B$ fits $H$, then $A\otimes B$ fits $G \boxtimes H$~\cite{Haemers1979}, where $A\otimes B$ is the Kronecker product of $A$ and $B$. Since $\rank(A\otimes B) = \rank(A)\rank(B)$ we have that $\gamma(G\boxtimes H) \le \gamma(G)\gamma(H)$, i.e.,  $\gamma(\cdot)$ is super-multiplicative w.r.t. the strong product. This guarantees the existence of the limit in the following definition.

\begin{definition}
 For a graph $G$, we define the \emph{fractional Haemers' minrank function} as
 \[
  \gamma_f(G) \eqdef \lim_{m \to \infty}\sqrt[m]{\gamma(G^m)}.
 \]
\end{definition}

The proof of the following upper bound follows that of Theorem~\ref{thm:meta-upb}, by virtue of Lemma~\ref{lem:Haemers}. 
\begin{theorem}[Fractional minrank upper bound]
 \label{thm:Haemers}
 For a source-channel pair $(G,H)$,
 \begin{equation}
 \label{eqn:Haemers_f}
  \Ir(H/G) \le \frac{\log \gamma_f(H)}{\log \gamma_f(G)}.
 \end{equation}
\end{theorem}


\subsection{Vertex-Transitive Channel Graph}

In this subsection, we study pairs $(G,H)$ where the channel graph $H$ is \emph{vertex-transitive}. Recall that $H$ is called vertex-transitive if for any two vertices $u,v$ of $H$, there exists some automorphism of $H$ that maps $u$ to $v$. Let $G,F$ be two graphs. We denote 
\[
\b(G,F)\eqdef \max\{|V(G')|\suchthat G' \text{ is an induced subgraph of } G \text{ and } G' \to F\}
\]
as the maximal number of vertices in an induced subgraph of $G$ that is homomorphic to $F$. 
The following two propositions investigate properties of $\b(G,F)$. In particular, Proposition~\ref{prop:subadd} states that this quantity is sub-multiplicative w.r.t. OR product in the first coordinate, and Proposition~\ref{prop:superadd} states that it is jointly super-multiplicative w.r.t. the OR product in both coordinates.

\begin{proposition}
\label{prop:subadd}
 For any graphs $G_1,G_2$, and $F$, 
 \[
 \max\{\a(G_1)\b(G_2,F),\, \a(G_2)\b(G_1,F)\} \le \b(G_1\vee G_2, F) \le \b(G_1,F)\b(G_2,F).
 \]
\end{proposition}

\begin{proof}
 For the lower bound, let $\widetilde{G}_1$ be an induced subgraph of $G_1$ that attains the maximum in $\b(G_1,F)$, i.e., there exists a homomorphism $\phi\suchthat \widetilde{G}_1 \to F$ and $|V(\widetilde{G}_1)| = \b(G_1,F)$. Let $\widetilde{G}_2$ be a maximum independent set of $G_2$. Consider the subgraph $G'$ of $G_1 \vee G_2$ induced by the vertices $V(\widetilde{G}_1) \times V(\widetilde{G}_2)$. We claim that the mapping $(g_1,g_2) \mapsto \phi(g_1)$ is a homomorphism from the induced subgraph $G'$ to $F$. This is because $\widetilde{G}_2$ is an independent set and thus $(g_1,g_2) \sim (g_1',g_2')$ if and only if $g_1\sim g_1'$. Since $\phi$ is a homomorphism, then $g_1 \sim g_1'$ implies $\phi(g_1) \sim \phi(g_1')$. Thus, $\b(G_1\vee G_2,F) \ge |V(G')| = |V(\widetilde{G}_1)|\cdot |V(\widetilde{G}_2)| = \b(G_1,F)\a(G_2)$. The other lower bound is obtained by swapping the roles of $G_1$ and $G_2$.
 
 For the upper bound, consider an arbitrary induced subgraph $G$ of $G_1 \vee G_2$ such that there exists a homomorphism $\phi\suchthat G \to F$. Each vertex of $G$ has two coordinates $(g_1,g_2)$. Suppose there are $L$ distinct $g_1$-coordinates in $G$, and let $ \{(g_{1i},g_{2i})\}_{i=1}^{L}$ be a subset of vertices of $G$ with distinct $g_1$ coordinate and potentially repeating $g_2$ coordinate (there can be more than one choice). Now let $\widetilde{G}_1$ be the subgraph of $G_1$ induced by the vertices $\{g_{1i}\}_{i=1}^{L}$. We claim that the mapping $g_{1i} \mapsto \phi(g_{1i},g_{2i})$ is a homomorphism $\widetilde{G}_1 \to F$. This is because whenever $g_{1i} \sim g_{1j}$, $(g_{1i},g_{2i}) \sim (g_{1j},g_{2j})$ by the definition of OR product, and thus $\phi(g_{1i},g_{2i}) \sim \phi(g_{1j},g_{2j})$. It follows that the number of distinct $g_1$ coordinates in $G$ is $L=|V(\widetilde{G}_1)| \le \b(G_1,F)$. With a similar analysis, the total number of distinct $g_2$ coordinates in $G$ is upper bounded by $\b(G_2,F)$. We conclude that $|V(G)| \le \b(G_1,F)\b(G_2,F)$ for any induced subgraph $G$ of $G_1\vee G_2$ that is homomorphic to $F$, establishing the upper bound. 
\end{proof}

\begin{proposition}
\label{prop:superadd}
 For any graphs $G_1, G_2, F_1$ and $F_2$, 
 \[
 \b(G_1,F_1)\b(G_2,F_2) \le \b(G_1 \vee G_2, F_1 \vee F_2) \le |V(G_1)||V(G_2)|.
 \]
\end{proposition}

\begin{proof}
Consider any induced subgraph $\widetilde{G}_1$ of $G_1$ such that there exists a homomorphism $\phi_1\suchthat \widetilde{G}_1 \to F_1$, and any induced subgraph $\widetilde{G}_2$ of $G_2$ such that there exists a homomorphism $\phi_2\suchthat \widetilde{G}_2 \to F_2$. Clearly, the mapping $(g_1,g_2) \mapsto (\phi_1(g_1),\phi_2(g_2))$ is a homomorphism $\widetilde{G}_1 \vee \widetilde{G}_2 \to F_1 \vee F_2$. Since $\widetilde{G}_1 \vee \widetilde{G}_2$ is an induced subgraph of $G_1 \vee G_2$, we have $|V(\widetilde{G}_1)||V(\widetilde{G}_2)| \le \b(G_1\vee G_2, F_1\vee F_2)$. This establishes the lower bound. The upper bound follows since any induced subgraph of $G_1 \vee G_2$ has no more than $|V(G_1)||V(G_2)|$ vertices.
\end{proof}

Proposition~\ref{prop:superadd} indicates that for any integer $k \ge 1$, $\b(G^{\vee km},F^{\vee m})\b(G^{\vee kn},F^{\vee n}) \le \b(G^{\vee k(m+n)}, F^{\vee (m+n)})$. This supermultiplicativity guarantees the existence of the limit in the following definition. 


\begin{definition}
\label{def:fracbeta}
 Let $\b_f^{(k)}(G,F) \eqdef \lim_{m \to \infty} \sqrt[km]{\b(G^{\vee km},F^{\vee m})}$. 
\end{definition}


The remainder of this subsection is devoted to establishing an implicit upper bound on information ratio, in terms of the quantity in Definition~\ref{def:fracbeta}. We first need the following two lemmas. 

\begin{lemma}
\label{lem:vtrans}
 If $G$ is vertex-transitive, then for any integer $k$, $G^{\vee k}$ is vertex-transitive.
\end{lemma}
The proof is straightforward and is omitted. 

\begin{lemma}[Proposition~4~\cite{Bondy--Hell1990}] Suppose $G,H$ are graphs and $H$ is vertex-transitive. If $G \to H$, then for any graph $F$,
\label{lem:beta}
 \begin{equation}
  \label{eqn:ratio}
  \frac{\b(G,F)}{|V(G)|} \ge \frac{\b(H,F)}{|V(H)|}.
 \end{equation}

\end{lemma}


\begin{remark}
 Taking $F = K_1$ and noting that $\b(G,K_1) = \a(G)$ for any graph $G$, Lemma~\ref{lem:beta} implies $\frac{\a(G)}{|V(G)|} \ge \frac{\a(H)}{|V(H)|}$ whenever $G \to H$. This is the so called \textit{no-homomorphism Lemma}. Taking $F = K_1,K_2,\ldots, K_{\chi(H)}$, Lemma~\ref{lem:beta} implies that if $G \to H$, then
 \[
  \frac{\b(G,K_s)}{|V(G)|} \ge \frac{\b(H,K_s)}{|V(H)|}
 \]
for every $1 \le s \le \chi(H)$. This recovers the dominance of normalized chromatic difference sequence due to Albertson and Collins~\cite[Theorem 2]{Albertson--Collins1985}. 
\end{remark}

Now we are ready to state the main result of this subsection, providing a necessary condition for the existence of a $(k,n)$ code for a pair with a vertex-transitive channel graph.
\begin{theorem}[{\bf Fractional homomorphism upper bound}]
\label{thm:beta}
 If a $(k,n)$ code exists for the pair $(G,H)$, and $H$ is vertex-transitive, then for any graph $F$,
 \begin{equation}
 \label{eqn:fracbeta}
  k\log \left(\frac{\b_f^{(k)}(\overline{G},F)}{|V(G)|}\right) \ge n \log \left(\frac{\b_f^{(n)}(\overline{H},F)}{|V(H)|}\right).
 \end{equation}
\end{theorem}

\begin{proof}
 By Lemma~\ref{lem:graph-hom}, if a $(k,n)$ code exists for $(G,H)$, then there exists a homomorphism $\overline{G}^{\vee k} \to \overline{H}^{\vee n}$. By Lemma~\ref{lem:m-ext}, for any integer $m$, there exists a homomorphism $\overline{G}^{\vee km} \to \overline{H}^{\vee nm}$. Note that by Lemma~\ref{lem:vtrans} and since vertex-transitiveness is invariant under  graph complement, $\overline{H}^{\vee nm}$ is vertex-transitive. Now, applying Lemma~\ref{lem:beta}, we have
 \[
  \frac{\b(\overline{G}^{\vee km},F^{\vee m})}{|V(G)|^{km}} \ge \frac{\b(\overline{H}^{\vee nm}, F^{\vee m})}{|V(H)|^{nm}}.
 \]
Taking the logarithm of both sides and letting $m \to \infty$ establishes~\eqref{eqn:fracbeta}.
\end{proof}

\section{Tightness In Bounds}

In this section, we discuss cases where our upper and lower bounds coincide. we then provide some examples in which the information ratio can either be determined exactly, or can be derived from functions of the participating graphs. 

\begin{theorem}
\label{thm:tight}
 The upper and lower bounds for $\Ir(H/G)$ coincide when any one of the following conditions are satisfied:
 \begin{enumerate}
  \item $\Theta(H) = \chif(H)$;
  \item $\Theta(G) = \chif(G)$;
  \item $\Theta(H) = \vartheta(H)$ and $\chif(G) = \vartheta(G)$.
 \end{enumerate}
\end{theorem}

\begin{remark}
\label{rmk:separation}
We see that whenever the graph capacity $\Theta$ coincides with $\chif$ for either one of the graphs, then the separation scheme is optimal regardless of the structure of the other graph. We will later see (in Proposition~\ref{prop:info-Ks}) that this situation has a nice ``physical'' interpretation: a graph whose capacity $\Theta$ coincides with $\chif$ is \emph{asymptotically equivalent} to an empty graph, in a sense that will be made precise in Section~\ref{sec:info-equiv}. Viewed this way, it is clear why separation is optimal in these two special cases. 
\end{remark}

\begin{example} We provide one example for each condition in Theorem~\ref{thm:tight}.
 \begin{enumerate}
  \item Let $H = \overline{K_s}$, then $\Ir(H/G) = \frac{\log s}{\log \chif(G)}$.
  \item Let $G = \overline{K_s}$, then $\Ir(H/G) = \frac{\log\Theta(H)}{\log s}$.
  \item Let $G$ be the Schl\"{a}fli graph (cf. Example~\ref{ex:Schlafli}) and $H = \overline{G}$. It is known that $\Theta(\overline{G}) = \vartheta(\overline{G}) = 3$ and $\chif(G) = \vartheta(G) = 9$~\cite{Schlafli-graph}. Therefore, we have $\Ir(\overline{G}/ G) = 0.5$.
 \end{enumerate}
\end{example}

We now consider some general examples where the graphs $G$ and $H$ are obtained by strong products or disjoint unions. The exact information ratio in these cases will be established by constructing a specific non-adjacency preserving mapping that attains the upper bound. 

\begin{example} For any graph $F$ and $m_1,m_2 \in \mathbbm{N}$,
\[
 \Ir(F^{m_1}/F^{m_2}) = \frac{m_1}{m_2}.
\]
Achievability follows by the identity mapping from $(F^{m_2})^{m_1}$ to $(F^{m_1})^{m_2}$. The converse follows from any of the upper bounds. 
\end{example}

\begin{remark}
 When $\Theta(F) \ne \chif(F)$, this shows that in contrast to the standard joint source-channel coding case, separation can be \emph{strictly suboptimal} in our  zero error setting. 
\end{remark}

\begin{example}For any graph $F$,
 \[
  \Ir(F+F/ F) = 1+\frac{1}{\log \chif(F)}.
 \]
For achievability, we describe a mapping from $F^{t+n}$ to $(F+F)^n$. Note that the latter graph is isomorphic to a disjoint union of $2^n$ identical graphs $F^n$, i.e., $(F+F)^n\cong 2^n F^n$. We map $F^n$ to $F^n$ via the identity mapping, and use the first $t$ coordinates to decide which of the $2^n$ copies to use. Clearly, adjacent vertices in $F^t$ can be mapped to the same copy of $F^n$. Given any clique cover of $F^t$, we can map each clique into a different copy of $F^n$. Therefore, $t\log \chif(F) = n$ is sufficient. Therefore, $\Ir(F+F/ F) \ge \frac{t+n}{n} = 1+\frac{1}{\log \chif(F)}$. For the converse, we apply the fractional chromatic number upper bound. Note that $\overline{F+F}$ is composed of two $\overline{F}$ subgraphs that are fully connected between them. So $2\chi_f(\overline{F})$ is necessary and sufficient to fractionally color $\overline{F+F}$. Thus, $\Ir(F+F/ F) \le \frac{\log \chif(F+F)}{\log \chif(F)} = \frac{\log (2\chif(F))}{\log \chif(F)} =1+\frac{1}{\log \chif(F)}$.
\end{example}

\begin{example}For any graph $F$,
 \[
  \Ir(F/F+ F) = \frac{\log\Theta(F)}{1+\log\Theta(F)}.
 \]
For achievability, consider a mapping from $(F+F)^k \cong 2^kF^k$ to $F^{t+k}$. We map $F^k$ to $F^k$ via the identity mapping, and use the first $t$ coordinates to indicate which of the $2^k$ copies is mapped to $F^t$. Clearly, the set of indices that determine the copy must be mapped to an independent set of size $2^k$ in $F^t$. So $k = t\log\Theta(F)$ is sufficient. Therefore, $\Ir(F/F+ F) \ge \frac{k}{t+k} = \frac{\log\Theta(F)}{1+\log\Theta(F)}$. For the converse, apply the capacity upper bound and note that $\Theta(F+F) = 2\Theta(F)$.
\end{example}

\begin{example}For any graph $F$ and $m_1,m_2 \in \mathbbm{N}$,
\[
 \Ir(F^{m_1}+F^{m_1}/F^{m_2}+F^{m_2}) = \begin{cases}
                                                      \frac{1+m_1\log\chif(F)}{1+m_2 \log\chif(F)} & \text{ if } m_1 \le m_2,\\
                                                      \frac{1+m_1 \log\Theta(F)}{1+m_2 \log\Theta(F)} & \text{ if } m_1 > m_2.
                                                     \end{cases}
\]
 For achievability, consider a mapping from $(F^{m_2}+F^{m_2})^k \cong 2^k F^{k m_2}$ to $(F^{m_1}+F^{m_1})^n \cong 2^n F^{n m_1}$. When $m_1 \le m_2$, let $\frac{m_1}{m_2} \le  \frac kn \le 1$. We map $2^k F^{nm_1}$ to the channel graph through the identity mapping. For the remaining coordinates, we map $F^{k m_2 - n m_1}$ to $2^{n-k}$ disjoint points by clique covering $F^{k m_2 - n m_1}$. This can be done if $(k m_2 - n m_1) \log\chif(F) = n-k$, or equivalently, $\frac{k}{n} = \frac{1+m_1\log\chif(F)}{1+m_2 \log\chif(F)}$. When $m_1 \ge m_2$, let $1 \le \frac kn \le \frac{m_1}{m_2}$. We map $2^n F^{k m_2}$ through the identity mapping to the channel graph. For the remaining coordinates, we map $2^{k-n}$ disjoint points to an independent set of $F^{n m_1 - km_2}$. This can be done if $k-n = (nm_1-km_2)\log\Theta(F)$, or equivalently, $\frac kn = \frac{1+m_1 \log\Theta(F)}{1+m_2 \log\Theta(F)}$. The converse follows from the upper bounds in terms of fractional chromatic number in case $m_1\le m_2$ and capacity in case $m_1 > m_2$.
\end{example}


\section{Homomorphic Equivalence and Cores}

In this section we further develop the homomorphism point of view, and show how it can be leveraged to simplify computation or approximation of the information ratio. Specifically, we discuss homomorphic equivalence and the concept of a graph \emph{core}~\cite{Hahn--Tardif1997}. We demonstrate that for the information ratio problem, we can limit our attention to pairs whose complementary graphs are cores. 

Two graphs $F_1$ and $F_2$ are called \emph{homomorphically equivalent} if there exist homomorphisms $F_1 \to F_2$ and $F_2 \to F_1$. When this is the case, we write $F_1 \leftrightarrow F_2$. It is easy to see that $\to$ is a partial order and that $\leftrightarrow$ is an equivalence relation on the set of finite graphs. The following simple theorem shows that the information ratio is monotonic w.r.t. the homomorphic order, and consequently that $\Ir(G/H)$ depends only on the homomorphic equivalence classes of $G$ and $H$. 

\begin{theorem}
\label{thm:eqIr}
Suppose $\overline{F_1} \to \overline{F_2}$. Then 
 \begin{enumerate}
  \item $\Ir(F_1/G) \leq \Ir(F_2/G)$ for any graph $G$.
  \item $\Ir(H/F_1) \geq \Ir(H/F_2)$ for any graph $H$.
 \end{enumerate}
Moreover, these bounds hold with equality if $\overline{F_1} \leftrightarrow \overline{F_2}$. 
\end{theorem}

\begin{proof}
by Lemma~\ref{lem:graph-hom}, for every $(k,n)$ code for $(G,F_1)$ there exists a homomorphism $\overline{G}^{\vee k} \to \overline{F_1}^{\vee n}$. Since $\overline{F_1} \to \overline{F_2}$ implies $\overline{F_1}^{\vee n} \to \overline{F_2}^{\vee n}$, then by applying the same mapping on each coordinate we have $\overline{G}^{\vee k} \to \overline{F_1}^{\vee n} \to \overline{F_2}^{\vee n}$. Recalling that the composition of two homomorphisms is also a homomorphism, we obtain $ \Ir(F_1/G) \le \Ir(F_2/G)$. The other case can be proved similarly by composing homomorphisms $\overline{F_1}^{\vee k} \to \overline{F_2}^{\vee k}$ and $\overline{F_2}^{\vee k} \to \overline{H}^{\vee n}$. Equality for $\overline{F_1} \leftrightarrow \overline{F_2}$ follows by swapping the role of $F_1$ and $F_2$. 
\end{proof}

\begin{remark}
\label{rmk:core}
The converse of Theorem~\ref{thm:eqIr} does not hold. Suppose that two graphs $F_1, F_2$ are such that $\Ir(F_1/G) = \Ir(F_2/G)$ \emph{for any} graph $G$ and $\Ir(H/F_1) = \Ir(H/F_2)$ \emph{for any} graph $H$. Then $\overline{F_1}$ and $\overline{F_2}$ are not necessarily homomorphically equivalent. For example, let $F_1=\overline{KG(6,2)}$ and $F_2=\overline{KG(12,4)}$. It is known that $\Theta(F_1) = \chif(F_1) = \Theta(F_2) = \chif(F_2) = 3$~\cite{Brouwer--Schrijver1979,Scheinerman--Ullman2011}. According to Theorem~\ref{thm:tight}, for any graphs $G$ and $H$, we have
\begin{align*}
\Ir(F_1/G) &= \frac{\log\Theta(F_1)}{\log \chif(G)} = \frac{\log\Theta(F_2)}{\log \chif(G)} = \Ir(F_2/G)\\
\Ir(H/F_1) &= \frac{\log\Theta(H)}{\log \chif(F_1)} = \frac{\log\Theta(H)}{\log \chif(F_2)} = \Ir(H/F_2).
\end{align*}
However, it is known that $KG(6,2) \to KG(12,4)$ and $KG(12,4) \not\to KG(6,2)$~\cite{Stahl1976}. In the next section, we correct this \textit{deficiency} of the homomorphic equivalence by introducing a slightly different notion of \textit{information equivalence}.
\end{remark}

Let us now use Theorem~\ref{thm:eqIr} in some specific examples. 
\begin{example}[Odd cycles]
 Let $C_{2n_1+1}$ and $C_{2n_2+1}$ be two odd cycles with $n_1 \ge n_2 \ge 1$. It is known that $\vartheta(\overline{C_{2n+1}}) = 1+1/\cos(\pi/(2n+1))$, and $\chif(\overline{C_{2n+1}}) = 2+1/n$~\cite{Lovasz1979}. It is easy to check that $C_{2n_1+1} \to C_{2n_2+1}$ and $K_2\to C_{2n+1}$. Using this yields 
\begin{align*}
  1 &\le \Ir(\overline{C_{2n_2+1}}/\overline{C_{2n_1+1}})
  \le \min\left\{\frac{\log(1+1/\cos(\pi/(2n_2+1)))}{\log(1+1/\cos(\pi/(2n_1+1)))},\, \frac{\log(2+1/n_2)}{\log(2+1/n_1)}\right\}
 \end{align*}
 and 
 \begin{align*}
  \frac{1}{\log (2+1/n_2)} &\le \Ir(\overline{C_{2n_1+1}}/\overline{C_{2n_2+1}})\\
  &\le \min\left\{\frac{\log(1+1/\cos(\pi/(2n_1+1)))}{\log(1+1/\cos(\pi/(2n_2+1)))},\, \frac{\log(2+1/n_1))}{\log(2+1/n_2))}\right\}.
 \end{align*}
\end{example}

\begin{example}[Odd wheels]
 Let $W_{2n_1+1}$ and $W_{2n_2+1}$ be two odd cycles with $n_1 \ge n_2 \ge 1$. One can check that $W_{2n_1+1} \to W_{2n_2+1}$, $K_3 \to W_{2n+1}$, $\vartheta(\overline{W_{2n+1}}) = 2+1/\cos(\pi/(2n+1))$, and $\chif(\overline{C_{2n+1}}) = 3+1/n$. Then the upper and lower bounds on $\Ir(\overline{W_{2n_2+1}}/\overline{W_{2n_1+1}})$  and $\Ir(\overline{W_{2n_1+1}}/\overline{W_{2n_2+1}})$ can be derived similarly as follows
 \begin{align*}
  1 &\le \Ir(\overline{W_{2n_2+1}}/\overline{W_{2n_1+1}})
  \le \min\left\{\frac{\log(2+1/\cos(\pi/(2n_2+1)))}{\log(2+1/\cos(\pi/(2n_1+1)))},\, \frac{\log(3+1/n_2)}{\log(3+1/n_1)}\right\}
 \end{align*}
 and 
 \begin{align*}
  \frac{\log 3}{\log (3+1/n_2)} &\le \Ir(\overline{W_{2n_1+1}}/\overline{W_{2n_2+1}})\\
  &\le \min\left\{\frac{\log(2+1/\cos(\pi/(2n_1+1)))}{\log(2+1/\cos(\pi/(2n_2+1)))},\, \frac{\log(3+1/n_1))}{\log(3+1/n_2))}\right\}.
 \end{align*}
\end{example}

\begin{example}[Kneser graphs]
 Let $KG(n_1,r_1)$ and $KG(n_2,r_2)$ be two Kneser graphs with $n_1 > 2r_1$ and $n_2 > 2r_2$. It is known that $\vartheta(\overline{KG(n,r)}) = \chif(\overline{KG(n,r)})= \frac{n}{r}$ and $\Theta(\overline{KG(n,r)}) \ge \lfloor \frac{n}{r}\rfloor$~\cite{Lovasz1979,Scheinerman--Ullman2011,Brouwer--Schrijver1979}. The latter implies that $K_m\to KG(n,r)^{\vee n}$ where $m$ is arbitrarily close to $\left(\lfloor \frac{n}{r}\rfloor\right) ^n$  for $n$ large enough. Then the upper and lower bounds on $\Ir(\overline{KG(n_1,r_1)}/\overline{KG(n_2,r_2)})$  are given as
 \[
  \frac{\log \lfloor n_1/r_1\rfloor}{\log (n_2/r_2)} \le \Ir(\overline{KG(n_1,r_1)}/\overline{KG(n_2,r_2)}) \le \frac{\log(n_1/r_1)}{\log (n_2/r_2)}.
 \]
\end{example}
\begin{example}
Recall that $(u_1,v_1)\sim (u_2,v_2)$ in the \textit{tensor product} $F_1\times F_2$, if $u_1 \sim u_2$ in $F_1$ and $v_1 \sim v_2$ in $F_2$. If $F_1 \to F_2$, then for any $G$ and $H$,
 \begin{align*}
  \Ir(\overline{F_1 + F_2}/G) &= \Ir(\overline{F_2}/G),\\
  \Ir(\overline{F_1 \times F_2}/G) &= \Ir(\overline{F_1}/G),\\
  \Ir(H/\overline{F_1+F_2}) &= \Ir(H/\overline{F_2}),\\
  \Ir(H/\overline{F_1\times F_2}) &= \Ir(H/\overline{F_1}).
 \end{align*}
\end{example}
To see this, note that if $F_1 \to F_2$, then $F_1 + F_2 \leftrightarrow F_2$ and $F_1 \times F_2\leftrightarrow F_1$.

The concept of homomorphic equivalence can also simplify the study of $\b_f^{(k)}(G,F)$ (cf. Definition~\ref{def:fracbeta}). 

\begin{proposition}
 If $F_1 \leftrightarrow F_2$, then  $\b_f^{(k)}(G,F_1) = \b_f^{(k)}(G,F_2)$ for any $G$ and $k$.
\end{proposition}

\begin{proof}
 Let $G'$ be any induced subgraph of $G$ such that $G' \to F_1$. Since $F_1 \to F_2$, we have $G' \to F_2$ and thus $\b(G,F_1) \le \b(G,F_2)$. By swapping the role of $F_1$ and $F_2$, we have $\b(G,F_2) \le \b(G,F_1)$. Now $F_1 \leftrightarrow F_2$ implies $F_1^{\vee m} \leftrightarrow F_2^{\vee m}$ for any integer $m$. Thus, $\b(G^{\vee km},F_1^{\vee m}) = \b(G^{\vee km},F_2^{\vee m})$ for any $m$. Taking $km$-th root on both sides and letting $m \to \infty$ completes the proof.
\end{proof}

Let $G$ and $H$ be graphs. Then $H$ is called a \emph{retract} of $G$ if there are homomorphisms $\rho : G \to H$ and $\gamma : H \to G$ such that $\rho \circ \gamma$ is an identity mapping on $H$. A graph $G$ is a \emph{core} if no proper subgraph of $G$ is a retract of $G$. A retract $H$ of $G$ is called a \emph{core of $G$} if it is a core. We denote the core of $G$ as $G^\bullet$. It is known that every finite graph has a core, and if $G_1$ and $G_2$ are cores of a graph $G$, then they are isomorphic. Moreover, a graph $G$ is a core if and only if every homomorphism $G\to G$ is an automorphism~\cite{Hahn--Tardif1997}. In the following, we provide a few canonical examples of cores and non-cores~\cite{Hahn--Tardif1997,Godsil--Royle2001}.


\begin{example}[Cores] The following graphs are cores:
\label{ex:cores}
 \begin{itemize}
  \item The complete graph $K_n$.
  \item The odd cycle $C_{2n+1}$.
  \item The odd wheel $W_{2n+1}$.
  \item Any $\chi$-critical graphs, that is, a graph for which the chromatic numbers of its proper subgraphs are strictly smaller than its chromatic number.
  \item The Kneser graph $KG(n,r)$ with $n > 2r$.
 \end{itemize}
\end{example}

\begin{example}[Non-cores] The following graphs are not cores.
\label{ex:noncore}
 \begin{itemize}
  \item The even cycle $C_{2n}$ (its core is $K_2$);
  \item The complete graph with one edge removed $K_m \setminus e$ (its core is $K_{m-1}$);
  \item A disjoint union of odd cycles $C_{2n+1} + C_{2m+1}$ (its core is $C_{2\min\{m,n\}+1}$).
 \end{itemize}
\end{example}

%

It is known that $F_1 \leftrightarrow F_2$ if and only if their cores $F_1^\bullet$ and $F_2^\bullet$ are isomorphic. This implies that up to isomorphism, a core is the unique graph with smallest number of vertices in the family of homomorphically equivalent graphs~\cite{Hahn--Tardif1997}. With these, we can now restate Theorem~\ref{thm:eqIr} in the language of cores. 


%


\begin{theorem}[Graph cores and information ratio]
\label{thm:core}
For any two graphs $G,H$, 
\begin{align*}
\Ir(H/G) = \Ir\left(\overline{(\overline{H})^\bullet}\,\Big/\,\overline{(\overline{G})^\bullet}\right).  
\end{align*}
\end{theorem}

It is now clear why cores are fundamental for information ratio problems. To compute the information ratio, it is sufficient to restrict our attention to pairs whose complementary graphs are cores. Since a core is the unique smallest graph in the equivalence class, this can sometimes simplify the calculations, as we now demonstrate. 

\begin{example}
\label{eg:core}
Let $F$ be a bipartite graph. Then for any $G$ and $H$, 
\[
 \Ir(\overline{F}/ G) = \frac{1}{\log \chif(G)}, \quad \Ir(H/ \overline{F}) = \log\Theta(H).
\]
To see this, note that the core of a bipartite graph is $K_2$. The claims then follow from Theorem~\ref{thm:core}, Proposition~\ref{prop:cap} and Proposition~\ref{prop:chif}.
\end{example}

\begin{example}
 Let $G = K_{m_1} + K_{m_2} + \cdots + K_{m_s}$ and $H = K_{M_1} + K_{M_2} + \cdots + K_{M_t}$ be disjoint unions of cliques. Then
 \[
  \Ir(H/G) = \frac{\log t}{\log s}.
 \]
To see this, note that $\overline{G}$ is a complete $s$-partite graph. The core of $\overline{G}$ is $K_s$. Similarly, the core of $\overline{H}$ is $K_t$. By Theorem~\ref{thm:core}, $\Ir(H/G) = \Ir(\overline{K_t}/\overline{K_s}) = \frac{\log t}{\log s}$.
\end{example}

Next we define the notion of the \emph{source/channel spectrum} of a graph, which is a characterization of the graph structure in terms of information ratios. In the following discussion, we fix an enumeration of the set of all non-isomorphic cores $\{\Gamma_i\}_{i=1}^\infty$.

\begin{definition}[Core spectra]
For any graph $G$, the sequence $\sigma_S(G)\eqdef \{\Ir(\overline{\Gamma_i}/G)\}_{i=1}^\infty$ is called the \emph{source spectrum} of $G$, and the sequence $\sigma_C(G)\eqdef\{\Ir(G/\overline{\Gamma_i})\}_{i=1}^\infty$ is called the \emph{channel spectrum} of $G$.
\end{definition}

As a straightforward application of Theorem~\ref{thm:core}, we have the following statement about core spectra.
\begin{proposition}
 Suppose $\overline{F_1} \leftrightarrow \overline{F_2}$. Then $F_1$ and $F_2$ have the same spectra, i.e., $\sigma_C(F_1) = \sigma_C(F_2)$ and $\sigma_S(F_1)=\sigma_S(F_2)$.
\end{proposition}

In the sequel, when we sum multiple spectra, or compare spectra, the operations will be meant element-wise in a natural way. The following is an equivalent formulation of Theorem~\ref{thm:Ir-prod}, Theorem~\ref{thm:Ir-invprod} and Lemma~\ref{cor:meta-upb}, in the language of spectra. 
\begin{corollary}
  The following relations hold:
  \begin{enumerate}
  \item $\sigma_S(G)\sigma_C(G)\leq 1$.
  \item $\sigma_C(G\boxtimes H) \geq \sigma_C(G) + \sigma_C(H)$.
  \item $\sigma_S(G\boxtimes H) \geq \frac{\sigma_S(G)\sigma_S(H)}{\sigma_S(G)+\sigma_S(G)}$.
  \end{enumerate}
\end{corollary}
We will get back to the spectra in the following sections.

\section{Information Equivalence and its Properties}
\label{sec:info-equiv}

Recall the deficiency of the homomorphic equivalence discussed in Remark~\ref{rmk:core}: whereas homomorphic equivalence (of the complements) implies unit information ratios, the reverse implication does not hold. In this section, we discuss two equivalence relations that are more natural for our problem, and study their properties. 

\subsection{Information Equivalence}

\begin{definition}
We say that two graphs $F_1$ and $F_2$ are \emph{information--equivalent} if $\Ir(F_1/ F_2) = \Ir(F_2/F_1) = 1$. When this is the case, we write $F_1 \Ieq F_2$.
\end{definition}

\begin{lemma}
 The relation $\Ieq$ is an equivalence relation on the set of all graphs. 
\end{lemma}

\begin{proof}
Reflexivity and antisymmetry are trivial. For transitivity, suppose that $F_1 \Ieq F_2$ and $F_2 \Ieq F_3$. Then by Lemma~\ref{lem:concat}, $\Ir(F_1/ F_3) \ge \Ir(F_1/ F_2) \Ir(F_2/ F_3) = 1$ and $\Ir(F_3/ F_1) \ge \Ir(F_3/ F_2) \Ir(F_2/F_1) = 1$. On the other hand, by Corollary~\ref{cor:meta-upb}, we have $\Ir(F_1/F_3) \Ir(F_3/F_1) \le 1$. Therefore, $\Ir(F_1/F_3) = \Ir(F_3/ F_1) = 1$ and thus $F_1 \Ieq F_3$.
\end{proof}

The follow proposition states that information equivalence is a coarsening of homomorphic equivalence. 

\begin{proposition}
$\overline{F_1} \leftrightarrow \overline{F_2}$ implies $F_1 \Ieq F_2$. The converse is not true in general.
\end{proposition}

\begin{proof}
The first claim follows from Theorem~\ref{thm:eqIr}. For the second claim, consider $F_1 = \overline{KG(6,2)}$ and $F_2 = \overline{KG(12,4)}$, for which $F_1 \Ieq F_2$ but $\overline{F_2} \not\to \overline{F_1}$ (cf. Remark~\ref{rmk:core}).
\end{proof}

The following theorem shows that information equivalence can be equivalently defined via the source/channel spectrum. 
\begin{theorem}[Information equivalence and graph spectra]
 \label{thm:ratio}
The following statements are equivalent:
 \begin{enumerate}
  \item $F_1 \Ieq F_2$.
  \item The source spectra are identical, i.e., $\sigma_C(F_1) = \sigma_C(F_2)$. 
  \item The channel spectra are identical, i.e., $\sigma_S(F_1) = \sigma_S(F_2)$. 
 \end{enumerate}
\end{theorem}


Let us first demonstrate the usefulness of this result. 
\begin{example}
 Let $s \ge 3$ be a positive integer. It is known that $\{KG(ms,m)\}_{m=1}^{\infty}$ are all non-isomorphic cores, and hence not homomorphically equivalent. However, for every $m$, $\overline{KG(ms,m)} \Ieq \overline{K_s}$. Applying Theorem~\ref{thm:ratio}, we have
 \begin{enumerate}
  \item $\Ir(\overline{KG(ms,m)}/G) = \Ir(\overline{K_s}/G) = \frac{\log s}{\log \chif(G)}$ for any graph $G$;
  \item $\Ir(H/\overline{KG(ms,m)}) = \Ir(H/\overline{K_s}) = \frac{\log\Theta(H)}{\log s}$ for any graph $H$.
 \end{enumerate}
\end{example}

\begin{proof}[Proof of Theorem~\ref{thm:ratio}]
 1) $\Rightarrow$ 2): Suppose that $F_1 \Ieq F_2$. By definition, $\Ir(F_1/F_2) = \Ir(F_2/F_1) = 1$. Then, for any graph $G$, 
\[
 \Ir(F_1/G) \ge \Ir(F_1/F_2) \Ir(F_2/G) = \Ir(F_2/G) \ge \Ir(F_2/F_1) \Ir(F_1/G) = \Ir( F_1/ G).
\] 
Hence, $\Ir(F_1/ G) = \Ir(F_2/G)$ for any graph $G$.

 1) $\Rightarrow$ 3): Similarly as above, for any graph $H$, 
\[
 \Ir(H/F_1) \ge \Ir(H/ F_2) \Ir(F_2/F_1) = \Ir(H/F_2) \ge \Ir(H/F_1) \Ir(F_1/F_2) = \Ir(H/F_1).
\] 
Hence, $\Ir(H/F_1) = \Ir(H/F_2)$ for any graph $H$.

2) $\Rightarrow$ 1): Suppose that $\Ir(F_1/ G) = \Ir(F_2/G)$ for any graph $G$. Taking $G = F_2$, we have $\Ir(F_1/F_2) = \Ir(F_2/F_2) = 1$. Taking $G = F_1$, we have $\Ir(F_2/F_1) = \Ir(F_1/F_1) = 1$. Hence, $F_1 \Ieq F_2$.

3) $\Rightarrow$ 1): Similarly as above, taking $H = F_1$ and $H= F_2$ respectively, we get $\Ir(F_2/F_1) = \Ir(F_1/F_1) = 1$.
\end{proof}

We now endow the space of information equivalence classes with a natural metric. Denote by $\Lc(F)$ the equivalence class of all graphs that are information--equivalent to the graph $F$.

\begin{definition}
 Let $G$ and $H$ be two graphs. Define the function 
 \[
 d(\Lc(G),\Lc(H))\eqdef -\log (\min\{\Ir(G/H),\,\Ir(H/G)\}).
 \]
We will write $d(G,H) \eqdef d(\Lc(G),\Lc(H))$ as a shorthand notation whenever convenient. 
\end{definition}

\begin{proposition}[{\bf Metric structure}]
Equipped with $d(\cdot,\cdot)$, the set of information equivalence classes forms a metric space. 
\end{proposition}

\begin{proof}
Since $\Ir(G/H)\Ir(H/G)\le 1$, we have $\min\{\Ir(G/H),\,\Ir(H/G)\} \le 1$. Thus, $d(G,H) \ge 0$. If $\Lc(G) = \Lc(H)$, i.e., $G \Ieq H$, then $\Ir(G/H) = \Ir(H/G) = 1$ and hence $d(G,H) = 0$. Conversely, if $d(G,H) = 0$, then $\min\{\Ir(G/H),\,\Ir(H/G)\} = 1$. It follows that $\Ir(G/H) \ge 1$ and $\Ir(H/G) \ge 1$. Together with the fact that $\Ir(G/H)\Ir(H/G) \le 1$, we have $\Ir(G/H) =\Ir(H/G) = 1$ and thus $\Lc(G) = \Lc(H)$. Symmetry follows immediately from definition. Finally, for subadditivity we have 
 \begin{align*}
  &d(\Lc(G),\Lc(H)) + d(\Lc(H),\Lc(F))\\
  &= -\log (\min\{\Ir(G/H),\,\Ir(H/G)\})-\log (\min\{\Ir(H/F),\,\Ir(F/H)\})\\
  &=-\log(\min\{\Ir(G/H),\,\Ir(H/G)\}\cdot \min\{\Ir(H/F),\,\Ir(F/H)\})\\
  &\ge -\log(\min\{\Ir(G/H)\Ir(H/F),\,\Ir(F/H)\Ir(H/G)\})\\
  &\ge -\log(\min\{\Ir(G/F),\,\Ir(F/G)\})\\
  &= d(\Lc(G),\Lc(F)).
 \end{align*}
\end{proof}

The next theorem states that $d(\cdot,\cdot)$ is contractive w.r.t. strong product.

\begin{theorem}[{\bf Contractivity  w.r.t. strong product}]
\label{thm:contraction}
For any three graphs $G,H,F$, 
\begin{align*}
d(G\boxtimes F, H\boxtimes F) \le d(G,H).  
\end{align*}
\end{theorem}

\begin{proof}
We need to show that 
\begin{equation}
\label{eqn:contraction}
\min\{\Ir(H\boxtimes F/G\boxtimes F),\,\Ir(G\boxtimes F/H\boxtimes F)\} \ge \min\{\Ir(H/G),\,\Ir(G/H)\}. 
\end{equation}

First, we show that 
\begin{align}\label{eqn:first_claim}
\Ir(H/G) \le 1 \quad\Rightarrow\quad  \Ir(H\boxtimes F/G\boxtimes F) \ge \Ir(H/G)
\end{align}
For any $(k,n)$ code for the pair $(G,H)$, we design a code for the pair $(G\boxtimes F, H\boxtimes F)$ as follows. We map $G^k$ to $H^n$ the same way as in the $(k,n)$ code for $(G,H)$. Since $\Ir(H/G) \le 1$, we must have $k \le n$. Thus, we can map $F^k$ to $F^n$ by $\phi(f^k) = (f^k,c^{n-k})$ for all $f^k \in F^k$, where $c^{n-k}$ is an arbitrary fixed vertex in $F^{n-k}$. Clearly both mappings are non-adjacency preserving. Therefore, a $(k,n)$ code for $(G\boxtimes F, H\boxtimes F)$ exists, and thus $\Ir(H\boxtimes F/G\boxtimes F) \ge \Ir(H/G)$.

Next, we show that 
\begin{align}\label{eqn:second_claim}
I(H/G) \ge 1 \quad\Rightarrow\quad \Ir(H\boxtimes F/G\boxtimes F) \ge 1.
\end{align}
If $I(H/G) \ge 1$, then there exists a sequence of $(k_i,n_i)$ codes for $(G,H)$ such that $k_i \le n_i, i = 1,2,3,\ldots$ and $\lim_{i \to \infty}\frac{k_i}{n_i} = 1$. For each code $(k_i,n_i)$, by the same construction as above, we can design a $(k_i,n_i)$ code for the pair $(G\boxtimes F, H\boxtimes F)$. Therefore, it must be that $\Ir(H\boxtimes F/G\boxtimes F) \ge 1$.

Let us now establish~\eqref{eqn:contraction}. Since $\Ir(H/G)\Ir(G/H) \le 1$, there are two possible cases: 1) Both $\Ir(H/G) \le 1$ and $\Ir(G/H)\le 1$. Then by~\eqref{eqn:first_claim} we have that $\Ir(H\boxtimes F/G\boxtimes F) \ge \Ir(H/G)$ and $\Ir(G\boxtimes F/H\boxtimes F) \ge \Ir(G/H)$, hence~\eqref{eqn:contraction} follows. 2) $\Ir(H/G) \le 1$ and $\Ir(G/H) \ge 1$ (or vice versa). Then by~\eqref{eqn:second_claim} we have that $\Ir(G\boxtimes F/H\boxtimes F) \ge 1$, and hence $\min\{\Ir(H\boxtimes F/G\boxtimes F),\,\Ir(G\boxtimes F/H\boxtimes F)\}=\Ir(H\boxtimes F/G\boxtimes F)$. Now since $\Ir(H/G) \le 1$, we have that $\Ir(H\boxtimes F/G\boxtimes F) \ge \Ir(H/G) = \min\{\Ir(H/G),\,\Ir(G/H)\}$ by virtue of~\eqref{eqn:first_claim}, and~\eqref{eqn:contraction} follows.  
\end{proof}


For homomorphic equivalence, we know that $G^\bullet \vee H^\bullet \leftrightarrow (G \vee H)^\bullet$ for any $G,H$\cite{Hahn--Tardif1997}. In other words, if $G_1 \leftrightarrow G_2$ and $H_1 \leftrightarrow H_2$, then $G_1 \vee H_1 \leftrightarrow G_1 \vee H_2 \leftrightarrow G_2 \vee H_2 \leftrightarrow G_2 \vee H_1$. We have the following analogous result for information equivalence.

\begin{proposition}
\label{prop:info-prod}
 If four graphs $G_1,G_2,H_1,H_2$ are such that $G_1 \Ieq G_2$ and $H_1 \Ieq H_2$, then
 \[
  G_1\boxtimes H_1 \Ieq G_1 \boxtimes H_2 \Ieq G_2 \boxtimes H_2 \Ieq G_2 \boxtimes H_1.
 \]
\end{proposition}

\begin{proof}
 By Theorem~\ref{thm:contraction}, $d(G_1\boxtimes H_1,G_1 \boxtimes H_2) \le d(H_1,H_2) =0$, $d(G_1 \boxtimes H_2, G_2 \boxtimes H_2) \le d(G_1,G_2) = 0$, and $d(G_2 \boxtimes H_2, G_2 \boxtimes H_1) \le d(H_2,H_1) = 0$. Therefore, $G_1\boxtimes H_1 \Ieq G_1 \boxtimes H_2$, $G_1 \boxtimes H_2 \Ieq G_2 \boxtimes H_2$, and $G_2 \boxtimes H_2 \Ieq G_2 \boxtimes H_1$. 
\end{proof}



\subsection{A Partial Order on Information Equivalence Classes}
\label{sec:partial}

In this subsection, we define a partial order on the set of information equivalence classes, which leads to interesting properties of the information ratio and several graph invariants. In the end, we provide an intuitive explanation why separation is optimal when $\Theta(G) = \chif(G)$ (cf. Remark~\ref{rmk:separation}).

\begin{definition}
\label{def:info-poset}
 For any two graphs $F_1$ and $F_2$, we write $\Lc(F_1) \preccurlyeq \Lc(F_2)$ if $\Ir(F_2/F_1) \ge 1$.
\end{definition}


\begin{lemma}[{\bf Partial information order}]
 The relation $\preccurlyeq$ is a partial order on the set of information equivalence classes.
\end{lemma}

\begin{proof}
Reflexivity:  For any equivalence class $\Lc(F)$, $\Ir(F/F) = 1$. Thus $\Lc(F) \le \Lc(F)$. Antisymmetry:  Suppose that $\Lc(F_1) \preccurlyeq \Lc(F_2)$ and $\Lc(F_2) \preccurlyeq \Lc(F_1)$, i.e., $\Ir(F_2/F_1) \ge 1$ and $\Ir(F_1/F_2) \ge 1$. Combined with the fact that $\Ir(F_2/F_1)\Ir(F_1/F_2) \le 1$, we have $\Ir(F_1/F_2) = \Ir(F_2/F_1) = 1$, i.e., $F_1 \Ieq F_2$. Thus $\Lc(F_1) = \Lc(F_2)$. Transitivity: Suppose that $\Lc(F_1) \preccurlyeq \Lc(F_2)$ and $\Lc(F_2) \preccurlyeq \Lc(F_3)$, i.e., $\Ir(F_2/F_1) \ge 1$ and $\Ir(F_3/F_2) \ge 1$. Then $\Ir(F_3/F_1) \ge \Ir(F_3/F_2)\Ir(F_2/F_1) \ge 1$. Thus $\Lc(F_1) \preccurlyeq \Lc(F_3)$.
\end{proof}

From Definition~\ref{def:info-poset}, if two graphs $F_1$ and $F_2$ are such that $\Ir(F_1/F_2) < 1$ and $\Ir(F_2/F_1) < 1$, then they are not comparable under the partial order $\preccurlyeq$. The following is an example of two incomparable graphs.

\begin{example}[Incomparable graphs]
Let $F_1 = C_5 \boxtimes C_5$ and $F_2 = \overline{K_6}$. We know that $\Theta(F_1) = \Theta(C_5)^2 = 5, \chif(F_1) = \chif(C_5)^2 = 6.25, \Theta(F_2)=  6$, and $\chif(F_2) = 6$. Thus,
 \[
  \Ir(F_1/F_2) \le \frac{\log\Theta(F_1)}{\log\Theta(F_2)} = \frac{\log 5}{\log 6} < 1 \quad \text{ and }\quad \Ir(F_2/F_1) \le \frac{\log \chif(F_2)}{\log \chif(F_1)} = \frac{\log 6}{\log 6.25} < 1.
 \]
\end{example}

\begin{example}[An infinite chain of partially ordered graphs]
The chain of empty graphs
\[
 \overline{K_1} \preccurlyeq \overline{K_2}  \preccurlyeq \overline{K_3} \preccurlyeq \cdots \preccurlyeq \overline{K_n} \preccurlyeq \cdots
\]
and the complement of odd cycles
\[
 \cdots \preccurlyeq \overline{C_{2n+1}} \preccurlyeq \cdots \preccurlyeq \overline{C_7} \preccurlyeq \overline{C_5} \preccurlyeq \overline{C_3}
\]
are two infinite chains of partially ordered graphs. More generally, for integers $p,q$ with $0 < q \le p$, let $K_{p/q}$ be the \emph{rational complete graph} with vertices $\{0,1,\ldots,p-1\}$ and edges $i \sim j$ if $q \le |i-j| \le p-q$. It is known that for $p/q \ge 2$, $K_{p/q} \to K_{p'/q'}$ if and only if $p/q \le p'/q'$~\cite[Theorem 6.3]{Hell--nesetril2004}. Moreover, for $p/q +1 \le p'/q'$, $\chi_f(K_{p/q}) < \chi_f(K_{p'/q'})$~\cite[Corollaries 6.4, 6.20]{Hell--nesetril2004}, which implies $\Lc(\overline{K_{p/q}}) \neq \Lc(\overline{K_{p'/q'}})$. Thus, 
$$\overline{K_{p_1/q_1}} \preccurlyeq \overline{K_{p_2/q_2}} \preccurlyeq \cdots \preccurlyeq \overline{K_{p_i/q_i}} \preccurlyeq \cdots$$ 
with $2 \le p_{i}/q_i +1 \le p_{i+1}/q_{i+1}$ for all $i$ is an infinite chain of partially ordered graphs. 
%
\end{example}

Recall that we compare spectra element-wise. 
\begin{theorem}[{\bf Information order and graph spectra}]
\label{thm:partial}
The following statements are equivalent:
\begin{enumerate}
 \item $\Lc(F_1) \preccurlyeq \Lc(F_2)$.
 \item $\sigma_S(F_2) \le \sigma_S(F_1)$.
 \item $\sigma_C(F_1) \le \sigma_C(F_2)$.
\end{enumerate}
\end{theorem}

\begin{proof}
1) $\Rightarrow$ 2) and 3): Suppose that $\Lc(F_1) \preccurlyeq \Lc(F_2)$, i.e., $\Ir(F_2/F_1) \ge 1$. Then $\Ir(F_2/G) \ge \Ir(F_2/F_1)\Ir(F_1/G) \ge \Ir(F_1/G)$ for any graph $G$, and $\Ir(H/F_1) \ge \Ir(H/F_2)\Ir(F_2/F_1) \ge \Ir(H/F_2)$ for any graph $H$.

2) $\Rightarrow$ 1): Suppose that $\Ir(F_1/G) \le \Ir(F_2/G)$ for any graph $G$. Taking $G = F_1$, we have $\Ir(F_2/F_1) \ge \Ir(F_1/F_1) = 1$, i.e., $\Lc(F_1) \preccurlyeq \Lc(F_2)$. 

3) $\Rightarrow$ 1): Suppose that $\Ir(H/F_1) \ge \Ir(H/F_2)$ for any graph $H$. Taking $H = F_2$, we have $\Ir(F_2/F_1) \ge \Ir(F_2/F_2) = 1$, i.e., $\Lc(F_1) \preccurlyeq \Lc(F_2)$. 
\end{proof}


\begin{example}
 Let $F_1$ be an induced subgraph of $F_2$. Then 
\begin{enumerate}
 \item $\sigma_S(F_2) \leq \sigma_S(F_1)$.
 \item $\sigma_C(F_1) \leq \sigma_C(F_2)$.
\end{enumerate}
To see this, note that the identity mapping from $F_1$ to $F_2$ is trivially non-adjacency preserving. Thus $\Ir(F_2/F_1) \ge 1$, i.e., $\Lc(F_1) \preccurlyeq \Lc(F_2)$ and the claims follow by virtue of Theorem~\ref{thm:partial}.
\end{example}

Recall Lemma~\ref{lem:hom-mono} and Lemma~\ref{lem:Haemers} that characterized several hom-monotone functions. The following proposition provides analogous statements for information equivalence. 
\begin{proposition}[Information-monotone functions]
\label{prop:info-mono}
 Suppose $\Lc(F_1) \preccurlyeq \Lc(F_2)$. Then
 \begin{enumerate}
  \item $\Theta(F_1) \le \Theta(F_1)$;
  \item $\vartheta(F_1) \le \vartheta(F_2)$;
  \item $\chif(F_1) \le \chif(F_2)$;
  \item $\gamma_f(F_1) \le \gamma_f(F_2)$.
 \end{enumerate}
\end{proposition}

\begin{proof}
 Follows directly from definition and Theorems~\ref{thm:meta-upb} and~\ref{thm:Haemers}.
\end{proof}

\begin{remark}
While the chromatic number is a hom-monotone function, it is not an information-monotone function; namely, it is not true in general that $\Lc(F_1) \preccurlyeq \Lc(F_2)$ implies $\bar{\chi}(F_1) \le \bar{\chi}(F_2)$. For example, let $F_1 = \overline{KG(6,2)}$ and $F_2 = \overline{KG(3,1)} = \overline{K_3}$. We know that $F_1 \Ieq F_2$ and thus $\Lc(F_1) \preccurlyeq \Lc(F_2)$. However $\bar{\chi}(F_1) = 4$, which is strictly greater than $\bar{\chi}(F_2) = 3$. Similarly, one can check that the independence number is not an information-monotone function either.  
\end{remark}

\begin{corollary}
\label{cor:info-mono}
 If $F_1 \Ieq F_2$, then
 \begin{enumerate}
  \item $\Theta(F_1) = \Theta(F_1)$;
  \item $\vartheta(F_1) = \vartheta(F_2)$;
  \item $\chif(F_1) = \chif(F_2)$;
  \item $\gamma_f(F_1) = \gamma_f(F_2)$.
 \end{enumerate}
\end{corollary}

%

Recall from Theorem~\ref{thm:tight} that if $\Theta(G) = \chif(G)$, then the separation scheme is optimal regardless of the structure of $H$. We now provide an intuitive explanation of this fact (see also Remark~\ref{rmk:separation}). It is well known that for any positive integer $s$, $\a(\overline{G}) = \chi(G) = s$ if and only if $G \leftrightarrow K_s$. The following proposition provides an asymptotic version of this statement. 
\begin{proposition}
\label{prop:info-Ks}
$G \Ieq \overline{K_s}$ if and only if $\Theta(G) = \chif(G) = s$.
\end{proposition}

\begin{proof}
 Suppose that $G \Ieq K_s$. By Corollary~\ref{cor:info-mono}, $\Theta(G) = \Theta(\overline{K_s}) = s$ and $\chif(G) = \chif(\overline{K_s}) = s$. On the other hand, assume that $\Theta(G) = \chif(G) = s$. Then $\Ir(G/\overline{K_s}) \ge \frac{\log\Theta(G)}{\log \chif(\overline{K_s})} = \frac{\log s}{\log s}= 1$ and $\Ir(\overline{K_s}/G) \ge \frac{\log\Theta(\overline{K_s})}{\log \chif(G)} = \frac{\log s}{\log s} = 1$.
\end{proof}

Recall that separation maps the source graph $G$ to an empty graph. Proposition~\ref{prop:info-Ks} states that when $\Theta(G) = \chif(G)$, the graph $G$ is information--equivalent to an empty graph, which is intuitively why separation incurs no loss in the this case, regardless of the structure of the channel graph $H$.


\subsection{Weak Information Equivalence}
In this subsection we introduce another information ratio based equivalence relation, whose equivalence classes also form a metric space in conjunction with a suitably defined distance function. This equivalence relation provides the most general conditions for equalities in Theorem~\ref{thm:Ir-prod} and Theorem~\ref{thm:Ir-invprod}.

\begin{definition}
 We say that two graphs $F_1$ and $F_2$ are \emph{weakly information--equivalent} if $\Ir(F_1/F_2)\Ir(F_2/F_1) = 1$. When this is the case, we write $F_1 \Ieq_w F_2$.
\end{definition}

\begin{lemma}
 The relation $\Ieq_w$ is an equivalence relation on the set of all graphs. 
\end{lemma}

\begin{proof}
Reflexivity and antisymmetry are trivial. For transitivity, suppose that $F_1 \Ieq_w F_2$ and $F_2 \Ieq_w F_3$, i.e., $\Ir(F_1/F_2)\Ir(F_2/F_1) = 1$ and $\Ir(F_2/F_3)\Ir(F_3/F_2) = 1$. Then $1 \ge \Ir(F_1/F_3)\Ir(F_3/F_1) \ge \Ir(F_1/F_2)\Ir(F_2/F_1)\Ir(F_3/F_2)\Ir(F_2/F_3) = 1$. Thus, $F_1 \Ieq_w F_3$.
\end{proof}

\begin{proposition}
 If $F_1 \Ieq F_2$ then $F_1 \Ieq_w F_2$. The converse is not true in general. 
\end{proposition}

\begin{proof}
 Clearly, $\Ir(F_1/F_2) = \Ir(F_2/F_1) = 1$ implies $\Ir(F_1/F_2)\Ir(F_2/F_1) = 1$. For the converse, let $F_1 = F_2^2$. Then $\Ir(F_1/F_2) = 2$ and $\Ir(F_2/F_1) = 1/2$. Therefore, $F_1$ is weakly information--equivalent to $F_2$, yet not information--equivalent to $F_2$.
\end{proof}

Recall that multiplying spectra by a scalar is done element-wise. 
\begin{theorem}[Weak information equivalence and graph spectra]
\label{thm:weak-equiv}
The following statements are equivalent:
 \begin{enumerate}
  \item $F_1 \Ieq_w F_2$.
  \item $\sigma_S(F_1) = \Ir(F_2/F_1)\cdot \sigma_S(F_2)$.
  \item $\sigma_C(F_1) = \Ir(F_1/F_2)\cdot \sigma_C(F_2)$.
 \end{enumerate}
\end{theorem}

\begin{proof}
1) $\Rightarrow$ 2): Suppose that $\Ir(F_1/F_2)\Ir(F_2/F_1) = 1$. Then, we have $\Ir(F_1/G) \ge \Ir(F_2/F_1)\Ir(G/F_2) \ge \Ir(F_2/F_1)\Ir(F_1/F_2)\Ir(G/F_1) = \Ir(G/F_1)$ for any graph $G$. Thus, every inequality in between should hold with equality. In particular, $\Ir(F_1/G) =\Ir(F_1/F_2)\Ir(F_2/G)$. 

1) $\Rightarrow$ 3): Similarly as above, suppose that $\Ir(F_1/F_2)\Ir(F_2/F_1) = 1$. We have $\Ir(H/F_1) \ge \Ir(F_2/F_1)\Ir(H/F_2) \ge \Ir(F_2/F_1)\Ir(F_1/F_2)\Ir(H/F_1) = \Ir(H/F_1)$ for any graph $H$. Thus, $\Ir(H/F_1) = \Ir(F_2/F_1)\Ir(H/F_1)$.

2) $\Rightarrow$ 1): Suppose that $\Ir(F_1/G) = \Ir(F_1/F_2)\Ir(F_2/G)$ for any graph $G$. Taking $G = F_1$, we get $1 = \Ir(F_1/F_1) = \Ir(F_1/F_2)\Ir(F_2/F_1)$. Thus, $F_1 \Ieq_w F_2$. 

3) $\Rightarrow$ 1): Similarly as above, taking $H = F_1$, we have $1= \Ir(F_1/F_1) =\Ir(F_2/F_1)\Ir(H/F_1) $. Thus $F_1 \Ieq_w F_2$.
\end{proof}


\begin{definition}
 Let $G$ and $H$ be two graphs. Define the function 
 \[d_w(\Lc(G),\Lc(H))\eqdef -\log (\Ir(G/H)\Ir(H/G)).\]
We will write $d_w(G,H) \eqdef d_w(\Lc(G),\Lc(H))$ as a shorthand notation whenever convenient. 
\end{definition}

\begin{proposition}[{\bf Metric structure}]
Equipped with $d_w(\cdot,\cdot)$, the set of weakly information equivalence classes forms a metric space. 
\end{proposition}

\begin{proof}
Since $\Ir(G/H)\Ir(H/G)\le 1$, we have $d_w(G,H) \ge 0$. By definition, $d_w(\Lc(G),\Lc(H)) = 0$ if and only if $\Ir(G/H)\Ir(H/G)= 1$, and hence $\Lc(G) = \Lc(H)$. Symmetry follows by definition. For subadditivity, 
 \begin{align*}
  d_w(G,H) + d_w(H,F) &= -\log (\Ir(G/H)\Ir(H/G))-\log (\Ir(H/F)\Ir(F/H))\\
  &= -\log([\Ir(G/H)\Ir(H/F)][\Ir(F/H)\Ir(H/G)])\\
  &\ge -\log(\Ir(G/F)\Ir(F/G))\\
  &= d_w(G,F).
 \end{align*}
\end{proof}

We now show how weak information equivalence is related to equalities in some of the information ratio inequalities derived in previous sections. 

\begin{theorem}
\label{thm:sum-ratio-tight}
 Let $G,H,F$ be graphs. Then
 \begin{equation}
 \label{eqn:ghf-ineq}
  \Ir(G\boxtimes H/ F) = \Ir(G/ F) + \Ir(H/ F),
 \end{equation}
provided that at least one pair is weakly information--equivalent, i.e., either $G \Ieq_w H$, or $F \Ieq_w G$, or $F \Ieq_w H$.
\end{theorem}

\begin{proof}
1) Suppose that $G \Ieq_w H$. We have
 \begin{align*}
  \Ir(G/F) &\stackrel{(a)}{\ge} \Ir(G/H)\Ir(H/H\boxtimes G) \Ir(H\boxtimes G/F)\\
  &\stackrel{(b)}{\ge} \Ir(G/H)\left(\frac{\Ir(H/G)}{1+\Ir(H/G)}\right)(\Ir(H/F)+\Ir(G/F))\\
  &\ge \Ir(G/H)\left(\frac{\Ir(H/G)}{1+\Ir(H/G)}\right)(\Ir(H/G)\Ir(G/F)+\Ir(G/F))\\
  & = \Ir(G/H)\Ir(H/G)\Ir(G/F)\\
  & \stackrel{(c)}{=} \Ir(G/F),
 \end{align*}
where $(a)$ follows by applying Lemma~\ref{lem:concat} twice, $(b)$ follows from~\eqref{eqn:prod2} of Theorem~\ref{thm:gh-prod} and~\eqref{eqn:Ir-prod} of Theorem~\ref{thm:Ir-prod}, and $(c)$ follows since $G \Ieq_w H$, i.e., $\Ir(H/G)\Ir(G/H) = 1$. Now every inequality in between has to hold with equality. In particular, we have~\eqref{eqn:ghf-ineq}.

2) Suppose that $F \Ieq_w G$. We have
\begin{align*}
 \Ir(H/G) &\ge \Ir(H/G\boxtimes H) \Ir(G\boxtimes H/F) \Ir(F/G)\\
 &\ge \left(\frac{\Ir(H/G)}{1+\Ir(H/G)}\right)(\Ir(G/F)+\Ir(H/F)) \Ir(F/G)\\
 &\stackrel{(d)}{=} \left(\frac{\Ir(H/G)}{1+\Ir(H/G)}\right)(\Ir(G/F)+\Ir(G/F)\Ir(H/G))\Ir(F/G)\\
 &=\Ir(H/G)\Ir(G/F)\Ir(F/G)\\
 &\stackrel{(e)}{=} \Ir(H/G),
\end{align*}
where $(d)$ follows since $F \Ieq_w G$ and thus $\Ir(H/F) = \Ir(G/F)\Ir(H/G)$ by Theorem~\ref{thm:weak-equiv}, and $(e)$ follows since $\Ir(G/F)\Ir(F/G) = 1$ when $F \Ieq_w G$. This establishes~\eqref{eqn:ghf-ineq}.
\end{proof}

\begin{theorem}
 \label{thm:inv-ratio-tight}
 Let $G,H,F$ be graphs. Then
 \begin{equation}
  \label{eqn:inv-ratio-tight}
  \Ir(F/G\boxtimes H) = \frac{\Ir(F/G)\Ir(F/H)}{\Ir(F/G)+\Ir(F/H)},
 \end{equation}
provided that at least one pair is weakly information--equivalent, i.e., either $G \Ieq_w H$, or $F \Ieq_w G$, or $F \Ieq_w H$.
\end{theorem}

\begin{proof}
 1) Suppose that $G \Ieq_w H$.  We have
 \begin{align*}
  \Ir(F/G) &\stackrel{(a)}{\ge} \Ir(F/G\boxtimes H)\Ir(G\boxtimes H/G) \\
  &\stackrel{(b)}{\ge} \left(\frac{\Ir(F/G)\Ir(F/H)}{\Ir(F/G)+\Ir(F/H)}\right)(1+\Ir(H/G))\\
  &\stackrel{(c)}{=} \left(\frac{\Ir(F/G)\Ir(F/H)}{\Ir(F/H)\Ir(H/G)+\Ir(F/H))}\right)(1+\Ir(H/G))\\
  &= \Ir(F/G),
 \end{align*}
where $(a)$ follows by Lemma~\ref{lem:concat}, $(b)$ follows by Theorem~\ref{thm:Ir-invprod} and~\eqref{eqn:prod1} of Theorem~\ref{thm:gh-prod}, and $(c)$ follows since $G \Ieq_w H$ and thus $\Ir(F/G) = \Ir(F/H)\Ir(H/G)$ by Theorem~\ref{thm:weak-equiv}. Now every inequality in between has to hold with equality. In particular, we have~\eqref{eqn:inv-ratio-tight}.

2) Suppose that $F \Ieq_w G$. We have
\begin{align*}
 \Ir(F/H) &\ge \Ir(F/G\boxtimes H)\Ir(G\boxtimes H/H)\\
 &\ge \left(\frac{\Ir(F/G)\Ir(F/H)}{\Ir(F/G)+\Ir(F/H)}\right)(1+\Ir(G/H))\\
 &\stackrel{(d)}{=} \left(\frac{\Ir(F/G)\Ir(F/H)}{\Ir(F/G)+\Ir(F/G)\Ir(G/H)}\right)(1+\Ir(G/H))\\
 &= \Ir(F/H),
\end{align*}
where $(d)$ follows since $F \Ieq_w G$ and thus $\Ir(F/H) = \Ir(F/G)\Ir(G/H)$ by Theorem~\ref{thm:weak-equiv}. This establishes~\eqref{eqn:inv-ratio-tight}.
\end{proof}

\section{Information--Critical Graphs}
\label{sec:critical}

In this section we introduce the notion of information--critical graphs, which are (informally) graphs that are ``minimal'' in the information ratio sense. We show that the complements of several known cores are information--critical. Below, we write $\sigma_S(F_1) < \sigma_S(F_2)$ to mean that the source spectrum of $F_1$ is strictly and uniformly smaller (on every coordinate) than that of $F_2$ (and similarly for $\sigma_C(\cdot)$). 

\begin{definition}
 A graph $F$ is  said to be \emph{information--critical} if there exists an edge $e \in E(F)$ such that both $\sigma_S(F\setminus e) < \sigma_S(F)$ and $\sigma_C(F\setminus e) > \sigma_C(F)$. 
\end{definition}
The following theorem provides equivalent characterizations of information--critical graphs. 
\begin{theorem}
  The following statements are equivalent:
  \begin{enumerate}
  \item $F$ is information--critical.
  \item $\sigma_S(F\setminus e) < \sigma_S(F)$ for some edge $e\in E(F)$. 
  \item $\sigma_C(F\setminus e) > \sigma_C(F)$ for some edge $e\in E(F)$. 
  \item $\Ir((F \setminus e)/F) > 1$ for some edge $e \in E(F)$.
  \end{enumerate}
\end{theorem}
\begin{proof}
1) $\Rightarrow$ 2) and 3): By definition. 

2) $\Rightarrow$ 4): 
  Suppose that there exists an edge $e \in E(F)$ such that $\sigma_S(F\setminus e) < \sigma_S(F)$. Then specifically for the channel graph $F\setminus e$, we have $\Ir((F \setminus e)/F) > \Ir((F \setminus e)/(F \setminus e)) = 1$. 

3) $\Rightarrow$ 4): 
 Suppose that there exists an edge $e \in E(F)$ such that $\sigma_C(F\setminus e) > \sigma_C(F)$. Then specifically for the source graph $F$, we have $\Ir((F \setminus e)/F) > \Ir(F/F) = 1$. 
  
  
  4) $\Rightarrow$ 1):
  Suppose that there exists an edge $e \in E(F)$ such that $\Ir((F \setminus e)/F) > 1$. Then for any graph $H$, $\Ir(H/F) \ge\Ir(H/(F\setminus e))\Ir((F\setminus e)/F) > \Ir(H/(F \setminus e))$.   For any graph $G$, $\Ir((F\setminus e)/G) \ge \Ir((F\setminus e)/F)\Ir(F/G) > \Ir(F/G)$. In other words, $\sigma_S(F\setminus e) < \sigma_S(F)$ and $\sigma_C(F\setminus e) > \sigma_C(F)$ for this edge $e$. By definition, $F$ is information--critical.
\end{proof}

Here are simple sufficient conditions for criticality. 
\begin{proposition}
\label{prop:critical}
 If there exists an edge $e \in E(F)$ such that $\Theta(F\setminus e) > \chif(F)$, then $F$ is information--critical. 
\end{proposition}
\begin{proof}
  Immediate by comparing lower and upper bounds.  
\end{proof}

\begin{corollary}\label{cor:tr_free}
  Suppose $\overline{F}$ is a connected triangle-free graph with at least three vertices, and $\chif(F) < 3$. Then $F$ is information--critical. 
\end{corollary}
\begin{proof}
  It is easy to see that we can always add an edge to $\overline{F}$ to create a triangle. Hence, we can remove an edge $e$ from $F$ such that $\Theta(F\setminus e) \geq \alpha(F\setminus e) \geq 3$. The claim follows by Proposition~\ref{prop:critical} since we assumed that $\chif(F) <3$.  
\end{proof}

\begin{example}
\label{ex:oddcycle}
Let $F = \overline{C_{2n+1}}$ be the complement of an odd cycle with $n \ge 2$. Then $F$ is information--critical. This follows from Corollary~\ref{cor:tr_free} since $C_{2n+1}$ is a connected triangle-free graph, and $\chif(\overline{C_{2n+1}}) = 2+\frac 1n < 3$. 
\end{example}

\begin{example}
 Let $F = \overline{W_{2n+1}}$ be the complement of an odd wheel. Then $F$ is information--critical. To see this, observe that the complement of a wheel is the disjoint union of the complement of a cycle with an isolated vertex, i.e., $\overline{W_{2n+1}} = \overline{C_{2n+1}} + \{0\}$. Label the vertices of $\overline{C_{2n+1}}$ by $\{1,2,\ldots,2n+1\}$, where $i$ and $(i+1 \mod 2n+1)$ are non-adjacent. By removing the edge $e$ connecting  $1$ and $3$, the vertices $\{0,1,2,3\}$ become an independent set of $\overline{W_{2n+1}}\setminus e$. Thus, $\Theta(\overline{W_{2n+1}}\setminus e) \ge 4 > \chif(\overline{W_{2n+1}}) = 3+\frac 1n$.
\end{example}

\begin{example}
 Let $F = \overline{KG(n,r)}$ be the complement of a Kneser graph, where $r$ does not divide $n$. Then $F$ is information--critical. To see this, recall that $\overline{KG(n,r)}$ is a graph whose vertices are $r$-subsets of $\{1,2,\ldots,n\}$, and where two vertices are adjacent if the intersection of the corresponding subsets is non-empty. Let $q = \lfloor n/r\rfloor$. For $1 \le i \le q$, let $v_i$ denote the vertex corresponds to the $r$-subset $\{(i-1)r+1,(i-1)r+2,\ldots,ir\}$. The vertex $v_q$ and the vertex corresponds to subset $\{n-r+1,n-r+2,\ldots,n\}$ (denoted by $v_0$) are adjacent, since $r$ does not divide $n$, which implies $\{(q-1)r+1,(q-1)r+2,\ldots,qr\}\cap \{n-r+1,n-r+2,\ldots,n\} \neq \emptyset$. By removing the edge $e$ between $v_q$ and $v_0$, the vertices set $\{v_0,v_1,v_2,\ldots,v_q\}$ become an independent set of $\overline{KG(n,r)}\setminus e$. Therefore, $\Theta(\overline{KG(n,r)}\setminus e) \ge \lfloor n/r\rfloor + 1 > \chif(\overline{KG(n,r)}) = n/r$.
\end{example}

\begin{example}
\label{ex:Mycielski}
The \textit{Mycielski construction}~\cite{Mycielski1955} takes a graph $F$ and creates a graph $M(F)$ with $2|V(F)|+1$ vertices, such that 1) $F$ is connected implies $M(F)$ is connected ; 2) $\alpha(\overline{M(F)}) = \alpha(\overline{F})$; and 3) $\chi_f(M(F)) = \chi_f(F) + \frac{1}{\chi_f(F)}$. Thus, using also Corollary~\ref{cor:tr_free} again, if $F$ is any graph such that $\overline{F}$ is connected and triangle-free, and $\chif(F) < \frac{3+\sqrt{5}}{2}$, then $\overline{M(\overline{F})}$ is information--critical. For example, $\overline{M(C_{2n+1})}$ is information--critical for any $n\geq 2$, and so is $\overline{M(M(C_{2n+1}))}$ for any $n\geq 7$. It is known that if $F$ is $\chi$-critical (cf. definition in Example~\ref{ex:cores}), then $M(F)$ is also $\chi$-critical, and hence a core~\cite{Che--Collins}. Therefore both $M(C_{2n+1})$ and $M(M(C_{2n+1}))$ are cores, providing yet another example of a core whose complement is information--critical. 
\end{example}

So far, every information--critical graph we saw is the complement of a core. In the next example, we provide two information--critical graphs whose complements are not cores.

\begin{example}
 Let $F = \overline{C_{2n}}$ be the complement of an even cycle with $n \ge 2$. Then $F$ is information--critical, since by Corollary~\ref{cor:tr_free}, $C_{2n}$ is a connected triangle-free graph and $\chif(F) = 2 < 3$.  Moreover, by Example~\ref{ex:Mycielski}, the Mycielski construction $\overline{M(C_{2n})}, n \ge 2$, is also information--critical. Note that $C_{2n}$ is not a core (cf. Example~\ref{ex:noncore}), and neither is $M(C_{2n})$. To see the latter, consider for example $M(C_{4})$ as illustrated in Figure~\ref{fig:mc4}. One can check that the mapping $f(1) = f(3) = 1, f(2) = f(4) = 2, f(1') = f(3') = 1',f(2')=f(4') = 2'$ and $f(0) = 0$ is a homomorphism $M(C_4) \to M(C_4)$, yet not an isomorphism. Thus $M(C_4)$ is not a core.
\begin{figure}[htbp]
 \small
 \centering
 \psfrag{1}[cc]{$1$}
 \psfrag{2}[cc]{$2$}
 \psfrag{3}[cc]{$3$}
 \psfrag{4}[cc]{$4$}
 \psfrag{0}[cc]{$0$}
 \psfrag{1'}[bc]{$1'$}
 \psfrag{2'}[bc]{$2'$}
 \psfrag{3'}[bc]{$3'$}
 \psfrag{4'}[bc]{$4'$}
 \includegraphics[scale = .6]{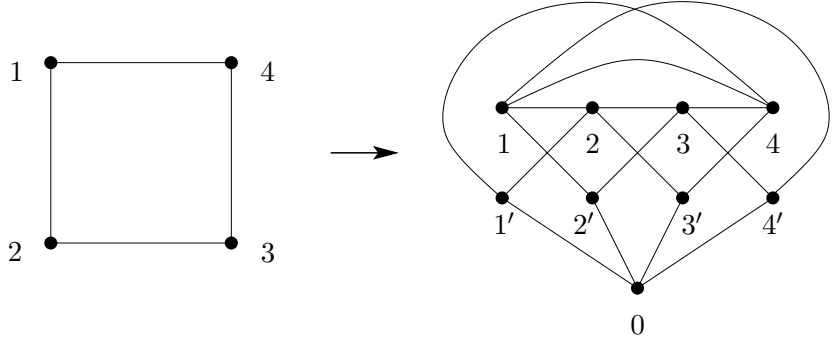}
 \caption{$C_4$ (left) and its Mycielski construction $M(C_4)$ (right). In $M(C_{4})$, the subgraph induced by vertices $\{1,2,3,4\}$ is isomorphic to $C_4$, $i' \sim 0$ for $i = 1,2,3,4$, and for each $i \sim j$ in $C_4$, there are two edges $i\sim j'$ and $i' \sim j$ in $M(C_4)$.}
 \label{fig:mc4}
\end{figure}
\end{example}


We now provide an example of a graph that is not information--critical. 
\begin{example}\label{ex:not_critical}
Consider $C_4$, the cycle over $4$ vertices. Removing any edge $e$ from $C_4$ results in a path over $4$ vertices, denoted by $P_4$. Clearly $\overline{C_4} \leftrightarrow \overline{P_4}$ and thus it holds that $\Ir(C_4/F) = \Ir(P_4/F) = \Ir((C_4\setminus e)/F)$ and $\Ir(F/C_4) = \Ir(F/P_4) = \Ir(F/(C_4\setminus e))$ for any graph $F$, and any edge $e$. Therefore $C_4$ is not information--critical (in a very strong sense, see also a short discussion in the subsequent  section).
\end{example}

\section{Open Problems}

In Section~\ref{sec:upper-bd}, we derived several upper bounds on the information ratio in terms of capacity, Lov{\'a}sz theta function, and fractional chromatic number (Theorem~\ref{thm:meta-upb}), Haemers' minrank function (Theorem~\ref{thm:Haemers}), and $\b(G,F)$ (Theorem~\ref{thm:beta}). It would be interesting to compare them. 

\begin{problem}
Can Theorem~\ref{thm:Haemers} strictly improve the other upper bounds in Theorem~\ref{thm:meta-upb}? 
\end{problem}
\begin{problem}
Can Theorem~\ref{thm:beta} strictly improve the other upper bounds in Theorem~\ref{thm:meta-upb} and Theorem~\ref{thm:Haemers} when the channel graph is vertex-transitive?  
\end{problem}

It is known that for the set of homomorphic equivalence classes, the partial order $\to$ forms a lattice~\cite{Gratzer1971}. For any two homomorphic equivalence classes $\Hc(G)$ and $\Hc(H)$, their least upper bound is the disjoint union $\Hc(G+H)$ and greatest lower bound is the tensor product $\Hc(G \times H)$~\cite[Theorem 2.37]{Hahn--Tardif1997}. In Section~\ref{sec:partial}, we defined a partial order $\preccurlyeq$ on the set of information equivalence classes (cf. Definition~\ref{def:info-poset})
\begin{problem}
 Does the information partial order form a lattice over the set of information equivalence classes? If so, what are the least upper bound and greatest lower bound of two information equivalence classes $\Lc(G)$ and $\Lc(H)$?
\end{problem}

Recall that the partial order induced by graph homomorphism (which refines the information partial order) contains an infinite antichain,  i.e., an infinite sequence of incomparable graphs. Thus we may ask:
\begin{problem}
Does the information partial order contain an infinite antichain? 
\end{problem}
We note in passing that even if the answer to the above question is negative, the information partial order is not a well partial order since it contains an infinite strictly decreasing chain: $\cdots \prec \overline{C_7} \prec \overline{C_5} \prec \overline{C_3}$, where $G \prec H$ if $\Lc(G) \preccurlyeq \Lc(H)$ and $\Lc(G) \neq \Lc(H)$. 

In Theorem~\ref{thm:contraction}, we showed that the function $d(G,H) \eqdef -\log\min\{\Ir(G/H),\,\Ir{H/G}\}$ is contractive w.r.t. strong product. We further conjecture it is also contractive w.r.t. disjoint union.
\begin{problem}
 Let $G,H$ be two graphs. Is it true that $d(G+ F, H+ F) \le d(G,H)$ for any graph $F$?
\end{problem}


In Section~\ref{sec:critical}, the complement of every core we were able to verify was information--critical. This leads to the following questions.

\begin{problem}
Suppose that $G$ has at least one edge, and $\overline{G}$ is a core. Then 
\begin{enumerate}[(1)]
\item Is $G$ information--critical?
\item Assume further that $G$ is not information--equivalent to $\overline{K_s}$ for any integer $s$. Is $G$ information--critical? 
\end{enumerate}
\end{problem}

It may be interesting to refine the notion of information--criticality. One way to do that is to let the edge being removed depend on the graph we compare to. Specifically, we can call a graph $F$ \emph{source--critical w.r.t. $H$} if there exists an edge $e$ such that $\Ir(H/(F\setminus e)) < \Ir(H/F)$, and \emph{channel--critical w.r.t. $G$} if there exists an edge $e$ such that $\Ir((F\setminus e)/G) > \Ir(F/G)$. Note that if a graph is information--critical, then it is both source--critical and channel--critical w.r.t. all graphs, but not necessarily vice-versa. Also, note that by Example~\ref{ex:not_critical}, there exists graphs that are not source--critical nor channel--critical w.r.t. any other graph. 

\begin{problem}
Consider the following questions:
\begin{enumerate}[(1)]
\item Are there graphs that are source--critical (resp. channel--critical) w.r.t. to all graphs, but not channel--critical (resp. source--critical) w.r.t. some graph? 
\item Find $G,H_1,H_2$ such that $G$ is source--critical w.r.t. $H_1$ but not w.r.t. $H_2$.  
\item Find $G_1,G_2,H$ such that $H$ is channel--critical w.r.t. $G_1$ but not w.r.t. $G_2$.  
\end{enumerate}
\end{problem}

\bibliographystyle{IEEEtran}
\bibliography{jscc}

\end{document}